\documentclass[11pt]{article}
\usepackage[english]{babel}
\usepackage[colorlinks,linkcolor=blue,citecolor=blue,urlcolor=blue]{hyperref}
\usepackage{amsmath,amssymb,graphicx,stmaryrd,cleveref,latexsym,fullpage}
\usepackage{amsopn,amsthm,upref}
\usepackage{microtype}
\usepackage{thmtools}

\usepackage{cite}
\usepackage{url}
\usepackage[caption=false]{subfig}

\usepackage{tikz}
\usetikzlibrary{cd}
\usetikzlibrary{decorations.pathreplacing}
\usepackage{ifthen}

\tikzset{vertex/.style={ draw , circle , fill , inner sep=0em , minimum size=0.3em}}
\tikzset{empty/.style={inner sep=0em, outer sep=0em, minimum size=0em}}

\newcommand{\Iff}{\textbf{if\textcompwordmark f} }

\newcommand{\tmem}[1]{{\em #1\/}}

\newcommand{\rbr}[1]{\left(#1\right)}

\newcommand{\abr}[1]{\left\langle#1\right\rangle}
\newcommand{\floorbr}[1]{\left\lfloor #1\right\rfloor}
\newcommand{\ceilbr}[1]{\left\lceil #1\right\rceil}
\newcommand{\fbr}[1]{\left\{#1\right\}}

\newcommand{\st}{\mathbf{st}}

\let\le\leqslant
\let\ge\geqslant
\let\phi\varphi

\newcommand{\eps}{\varepsilon}

\theoremstyle{plain}
\newtheorem{lemma}{Lemma}
\newtheorem{proposition}{Proposition}

\newtheorem*{problem}{Open problem}

\theoremstyle{definition}

\newtheorem*{definition}{Definition}

\theoremstyle{remark}
\newtheorem{remark}{Remark}

\newtheorem*{remark*}{Remark}

\newcommand{\ZZ}{\mathbb{Z}}
\newcommand{\NN}{\mathbb{N}}

\newcommand{\abs}[1]{\left\lvert #1 \right\rvert}
\newcommand{\norm}[1]{\left\lVert#1\right\rVert}

\newcommand{\HX}{H_X}
\newcommand{\HZ}{H_Z}
\newcommand{\CX}{\mathcal{C}_X}
\newcommand{\CZ}{\mathcal{C}_Z}
\newcommand{\dX}{d_X}
\newcommand{\dZ}{d_Z}
\newcommand{\wt}{\mathbf{wt}}

\newcommand{\T}{*}
\newcommand{\id}{\mathrm{id}}

\newcommand{\zm}{0}

\newcommand{\cA}{\mathcal{A}}
\newcommand{\cB}{\mathcal{B}}
\newcommand{\cC}{\mathcal{C}}

\newcommand{\cG}{\mathcal{G}}

\newcommand{\cQ}{\mathcal{Q}}

\newcommand{\cL}{\mathcal{L}}

\newcommand{\cF}{\mathcal{F}}

\newcommand{\sC}{{\mathbf{C}}}

\newcommand{\cay}{\mathrm{Cay}}

\newcommand{\Z}{\mathbb{Z}}
\newcommand{\F}{\mathbb{F}}

\DeclareMathOperator{\rk}{rk}
\DeclareMathOperator{\im}{im}
\DeclareMathOperator{\supp}{supp}

\newcommand{\Gr}{{\sf Gr}}
\newcommand{\PP}{\mathsf{P}}
\newcommand{\qbin}[2]{\binom{#1}{#2}_{\! q}}

\begin{document}

\title{Two-sided Robustly Testable Codes}
\author{Gleb~Kalachev and Pavel~Panteleev\thanks{Gleb~Kalachev and Pavel~Panteleev are with the Faculty of Mechanics and Mathematics, Moscow State University, Moscow, Russia.%
}}

\setcounter{page}{1}
\maketitle

\begin{abstract}
  We show that the tensor product of two random linear codes is robustly testable with high probability. This implies that one can obtain pairs of linear codes such that their product and the product of their dual codes are simultaneously robustly testable. Such two-sided robustly testable codes (with a much weaker form of robustness) were the~key ingredient in the recent constructions of asymptotically good quantum LDPC codes, which ensured their linear minimum distance. We hope that the existence of such codes with a stronger form of robustness, shown here, can be used to simplify the proofs and provide better distance bounds in these constructions. We also give new very simple examples of non-robustly testable codes. We show that if the parity-checks of two codes are mutually orthogonal, then their product is not robustly testable. In particular, this implies that the product of a~code with its dual can never be robustly testable. We also study a~property of a~collection of linear codes called product-expansion, which can be viewed as a~coboundary expansion of the cochain complex naturally associated with the product of these codes. We show that this property is related with the robust testability and the agreement testability of the products of codes. 
\end{abstract}

\section{Introduction}

A locally testable code~\cite{LTC:2006} is an error correcting code $\cC\subseteq \F_q^n$ equipped with an~efficient non-deterministic test which reads a~very small number of symbols from $x\in\F_q^n$ and allows to estimate the  distance from $x$ to $\cC$. Formally, we say that a~linear code $\cC\subseteq \F_q^n$ is (\emph{strongly}) \emph{locally testable} with \emph{soundness} $s$ and \emph{locality} $w$  if it has a~parity-check matrix $H$ with rows of weight at most $w$ such that for any vector $x \in \mathbb{F}_q^n$ we have
\begin{equation}\label{eqn:LTC}
   s \cdot \delta(x, \mathcal{C}) \le  \norm{H x}, 
\end{equation}
where $\delta(x, \mathcal{C}) := \min_{c\in\cC} \delta(x,c)$, and we denote respectively by $\delta(\cdot,\cdot)$ and $\norm{\cdot}$ the Hamming distance and the Hamming weight, both normalized by the length of the corresponding vectors. A~code~$\cC$ satisfying the above definition can be equipped with the following proximity test. Given $x\in \F_q^n$ we can pick uniformly at random a row from $H$ and check whether $x$ satisfies the corresponding linear equation. From~condition~(\ref{eqn:LTC}) it follows that the \emph{smaller} the rejection probability $\mathrm{rej}_H(x) = \norm{Hx}$ of this test (the right-hand-side) the \emph{closer} $x$ is to~$\cC$ (the left-hand-side).
Locally testable codes are very important in computer science since they can be viewed as the~combinatorial core of Probabilistically Checkable Proofs (PCPs)~\cite{LTC:2006, Dinur:2006}.

Robustly testable codes were introduced by Ben-Sasson and Sudan in~\cite{Ben-Sasson:2006} as a way to obtain locally testable codes using tensor products.
Recall that given two linear codes $\cC_i\subseteq \F_q^{n}$, $i\in [2]$, we can define the corresponding (\emph{tensor}) \emph{product code} 
\[
\cC_1\otimes\cC_2 := \{c\in \F_q^{n_1\times n_2} \mid \forall i\in [n_1]\ \forall j\in [n_2]\colon c(\cdot, j) \in \cC_1, c(i,\cdot) \in \cC_2 \},
\]
where $\F_q^{n_1\times n_2}$ is the space of $n_1\times n_2$ matrices over $\F_q$, and for a~matrix $c$ we denote by $c(\cdot, j)$ and $c(i,\cdot)$ its $j$-th column and $i$-th row, respectively. In words, the code~$\cC_1\otimes\cC_2$ consists of all matrices where each column is from $\cC_1$ and each row is from~$\cC_2$. We say that $\cC_1 \otimes \cC_2$ is \emph{$\rho$-robustly testable} if for each $x\in \F_q^{n_1\times n_2}$ we have:
\begin{equation}\label{eqn:robust-df}
\rho\cdot \delta(x, \cC_1\otimes\cC_2) \le \frac12\left(\delta(x, \cC_1\otimes \F_q^{n_2}) + \delta(x, \F_q^{n_1}\otimes \cC_2)\right).
\end{equation}
We can interpret~$\rho$ as the parameter controlling the robustness of the following natural proximity test for $\cC_1\otimes\cC_2$. Given a~matrix $x\in \F_q^{n_1\times n_2}$ we can try to check whether it is close to the code $\cC_1\otimes\cC_2$ by picking uniformly at random one of its columns or rows and looking at the distance to $\cC_1$ and~$\cC_2$, respectively. If $\cC_1\otimes\cC_2$ is $\rho$-robustly testable, then from~(\ref{eqn:robust-df}) it follows that the \emph{smaller} the average distance we get using this test (the right-hand-side) the \emph{closer} $x$ is to $\cC_1\otimes\cC_2$ (the left-hand-side).

Given a~finite field $\F_q$ and two numbers $R_1\in (0,1)$ and $R_2\in (0,1)$, it is natural to ask whether for a~random pair of codes of length $n$ and dimensions $R_1 n$ and $R_2 n$, their product is robustly testable for some $\rho$ with high probability as $n\to\infty$. Our main result (Theorem~\ref{th:rand-expanding}) implies that this is indeed the case for every $\F_q$, $R_1\in (0,1)$, and $R_2\in (0,1)$. However, in this paper we prefer to prove this result using a different form of robust testability called \emph{product-expansion}, which we discuss later. Note that the product-expansion is in turn  closely related with yet another form of robust testability called \emph{agreement testability}, which is studied by Dinur and Kaufman in~\cite{Dinur:2017} (see also~\cite[Definition~2.8]{Dinur:stoc2022}). 
Let us recall this definition. A~code $\cC_1\otimes \cC_2$ is called \emph{$\rho$-agreement testable} if 
for every $c_1\in \cC^{(1)}:=\cC_1\otimes \F_q^{n}$ and $c_2\in  \cC^{(2)}:=\F_q^{n}\otimes \cC_2$ there exists $c\in\cC^{(1)}\cap\cC^{(2)} = \cC_1 \otimes \cC_2$ such that 
\begin{equation}\label{eqn:agreement}
    \rho\cdot (\norm{c_1 - c}_1 + \norm{c_2 - c}_2) \le \norm{c_1 - c_2},
\end{equation}
where $\norm{\cdot}_1$ and $\norm{\cdot}_2$ are respectively the fraction of non-zero columns and rows in a~matrix. In words, this definition means that the \emph{smaller} the amount of ``disagreement'' between $c_1\in \cC^{(1)}$ and $c_2\in\cC^{(2)}$ (the right-hand-side) the \emph{less} the fraction of columns in $c_1$ and rows in $c_2$ one has to change to get to the ``total agreement'' (the left-hand-side), i.e., to obtain a~codeword  $c\in\cC^{(1)}\cap\cC^{(2)}$.

Robust testability is usually studied for products of codes $\cC\subseteq\F_q^{n}$ of non-vanishing \emph{relative distance} $\delta(\cC) := \min_{c\ne c'\in \cC_i}\delta(c,c')$ and \emph{rate} $R(\cC):= \frac1n \dim\cC$ as $n\to\infty$. It is well-known~\cite{Valiant:2005, Goldreich:2012} that \emph{not} for every pair of codes $\cC_1, \cC_2$ of non-vanishing relative distance and rate the product code $\cC_1\otimes \cC_2$ is necessary robustly testable. It is interesting that quantum CSS codes also can be used to provide examples of non-robustly testable codes (see~Remark~\ref{rm:exam-non-robust}). Recall that a~quantum CSS code~~\cite{CSS:1996, CSS2:1996} $\cQ$ of~\emph{dimension~$k$} is defined by a~pair of classical linear codes $\CX, \CZ \subseteq \mathbb{F}_q^n$ such that $\CZ^{\perp} \subseteq
\CX$, and $k = \dim \CX / \CZ^{\perp}$. Its~\emph{minimum distance $d$} is defined as
$\min (\dX, \dZ)$, where $\dX$ and $\dZ$ are the minimal Hamming weights of
the~vectors from $\CX \setminus \CZ^{\perp}$ and $\CZ \setminus \CX^{\perp}$,
respectively. The codes $\CX$, $\CZ$ are usually represented respectively by parity-check
matrices~$\HX$, $\HZ$, and the condition
$\CZ^{\perp} \subseteq \CX$ is equivalent to $\HX \HZ^{\T} = \zm$ (commutativity condition), where $\HZ^*$ is the transpose of $\HZ$. As we will see later in Remark~\ref{rm:exam-non-robust}, this commutativity condition implies that the code $\cC_X\otimes \cC_Z$ is not robustly testable. In particular, for every linear code $\cC$ the code $\cC \otimes \cC^\perp$ is not robustly testable. 

It is also interesting to note that this observation\footnote{Recently we learned that a~similar observation is also independently made by Dinur and Vidick~\cite{Dinur:2022ex}.} shows that the Polishchuk-Spielman construction~\cite{Polishchuk:1994} of robustly testable codes based on the tensor product of two Reed-Solomon codes is essentially optimal. Indeed, it works only for pairs of Reed-Solomon codes of rates $R_1$ and $R_2$ provided that $R_1 + R_2 < 1 - \eps$, where $\eps > 0$ is some arbitrary small constant. However, it is well-known that the dual code $(\mathrm{RS}_q^{k})^\perp$ to the~Reed-Solomon code~$\mathrm{RS}_q^{k} := \{(f(\alpha))_{\alpha\in\F_q} \mid f \in \F_q[x], \deg f < k\}$ coincides with $\mathrm{RS}_q^{q-k}$. Hence if $k_2 \ge q - k_1$ then $(\mathrm{RS}_q^{k_1})^\perp = \mathrm{RS}_q^{q - k_1} \subseteq \mathrm{RS}_q^{k_2}$, and   
the code $\mathrm{RS}_q^{k_1} \otimes \mathrm{RS}_q^{k_2}$ is not robustly testable if
\[
R_1 + R_2 = k_1/n + k_2/n \ge 1
\]
as the length $n=q$ of the component Reed-Solomon codes increases.   

We already mentioned that robustly testable codes can be used to construct LTCs. But, as it turns out, they are also useful to construct  low-density parity-check quantum CSS codes. Recall that if the parity-check matrices in a family of codes (classical or quantum) are sparse (e.g., the weight of each row is bounded above by some constant $w$), then such codes are called \emph{low-density parity-check} (LDPC) codes~\cite{gallager1963, Mackay:2004}.
It is clear that LTCs with constant locality $w$ are LDPC codes.  
Robustly testable codes appeared in different forms in the first examples of asymptotically good\footnote{Recall that an~infinite collection of (classical or quantum) codes  is called (\emph{asymptotically}) \emph{good} if there exist a~number $\eps>0$ such that both the rate and the relative minimum distance of these codes are bounded below by~$\eps$.} classical LTCs~\cite{Dinur:stoc2022, Panteleev&Kalachev:stoc2022} and quantum LDPC (qLDPC) codes~\cite{Panteleev&Kalachev:stoc2022}. In~\cite{Dinur:stoc2022} Dinur \emph{et al.} use agreement testable codes~\cite{Dinur:2017}, which, as they show, are equivalent to the robustly testable codes if the codes have non-vanishing relative distance. In~\cite{Panteleev&Kalachev:stoc2022} the authors use the product-expansion property, which, as we will see later, is equivalent (in its strongest form) to the agreement testability of a~product of two codes, and hence also to the robust testability for products of codes of linear minimal distances. However, in the case of products of more than \emph{two} codes the situation is much more interesting. On the one hand, it is not hard to show that product-expansion always implies  robust testability. However, when we have a~product of more than two codes, it can be shown that robust testability no longer implies product-expansion~\cite{Kalachev:2023}. Note that the term \emph{product-expansion} was originally motivated by the fact that it is a form of coboundary \emph{expansion}~\cite{Linial:2006,Gromov:2010} in the~cochain complex naturally associated with the \emph{product} of several codes (see Appendix~\ref{app:prod-exp}).  Such complexes appeared in~\cite{Panteleev&Kalachev:stoc2022} for the analysis of local neighborhoods in a~two-dimensional complex corresponding to a~product of two Tanner codes.

\begin{remark}
In retrospect one can see that good LTCs and qLDPC codes from~\cite{Dinur:stoc2022,Panteleev&Kalachev:stoc2022} can be obtained respectively as the second and the first homology of the $2$-dimensional tensor product complex\footnote{The general definition of the \emph{tensor product complex $\cA \otimes_R \cB$} over an~arbitrary ring~$R$ can be found in~\cite[p.~7]{Brown:1982}.} $\cA \otimes_G \cB := \cA \otimes_R \cB$ over a~group algebra 
$R=\F_qG$ obtained from $1$-dimensional chain complexes $\cA$ and $\cB$, related to some linear codes. This general approach was first proposed as a~way to obtain qLDPC codes by the authors in~\cite{Panteleev&Kalachev:2021}, and later, in a~more general form, by Breuckmann and Eberhardt~\cite{Breuckmann:balanced:2021, Breuckmann:2021}. It is interesting that the algebraic construction $\cA\otimes_G \cB$, which is often used in the context of group cohomology~\cite[p.~55]{Brown:1982}, has several different geometrical interpretations. Depending on a situation, it may be viewed as a~\emph{lift} (i.e., a~$|G|$-fold covering) of the~\emph{product} of two topological spaces, the~\emph{balanced product} of two topological spaces with an~action of a~group~$G$, or as a~\emph{fiber bundle}.   
This explains the plethora of different names used in the literature to describe such qLDPC codes~\cite{Panteleev&Kalachev:2021,Panteleev&Kalachev:stoc2022, Breuckmann:balanced:2021, Breuckmann:2021, Hastings:2021:fiber}. Also note that one can automatically apply this tensor product construction to graphs, hypergraphs, and even to general incidence structures since they all can be viewed as $1$-dimensional complexes~\cite{Breuckmann:balanced:2021,Panteleev&Kalachev:stoc2022}. 
\end{remark}

It was shown in~\cite{Panteleev&Kalachev:2021} that combining LDPC codes of constant rate
(e.g., Tanner codes) the construction $\cA \otimes_G \cB$ can give qLDPC codes of constant rate%
\footnote{A~similar observation (without a~proof) was also made in~\cite{Breuckmann:balanced:2021}.}. 
However, it was unclear at that time how to combine two codes to also get linear minimum distance. This goal was achieved in~\cite{Panteleev&Kalachev:stoc2022} by introducing a~new construction, called \emph{expander lifted product codes}, that combines two Tanner codes such that their local codes are two-sided robustly testable (in a~weak form). This gave a~first family of good qLDPC codes. In the current paper, we show that such two-sided robustly testable codes also exist in the strong form described below.

As we already mentioned before, the code $\cC \otimes \cC^\perp$ is never robustly testable. This fact suggests that a~straightforward approach to build good qLDPC codes using the tensor product $\cA \otimes_G \cB$ of two good classical codes may fail despite that we have some numerical evidence~\cite[Example~4]{Panteleev&Kalachev:2021} that large distances in the component codes $\cA$ and $\cB$ may lead to large distances in the quantum codes $\cA \otimes_G \cB$, as it was in the original hypergraph product construction~\cite{Tillich&Zemor:2009}. Take, for example, the explicit family of qLDPC codes $\cA \otimes_G \cA^*$ conjectured to be good by Breuckmann and Eberhardt in~\cite{Breuckmann:balanced:2021}, where $\cA$ is the Tanner code defined on a~Cayley graph $\cay(G,S)$ with a~local code $\cL$. Since these codes are similar%
\footnote{Despite the similarity, such codes have significantly different properties. For example, they do not look locally as product codes, which is one of the key properties used in the proof from~\cite{Panteleev&Kalachev:stoc2022}.} 
to the lifted product codes, it is tempting to adapt the proof from~\cite{Panteleev&Kalachev:stoc2022} to this case. However, for this proof to work, the product code $\cL \otimes \cL^\perp$ should be robustly testable, which is \emph{never} the case. Thus we have an~interesting open problem. 

\begin{problem}
    Whether the observation that $\cL\otimes \cL^\perp$ is never robustly testable leads to a~counterexample to the conjecture of Breuckmann and Eberhardt from~\cite{Breuckmann:balanced:2021} or this explicit family of qLDPC codes may still be shown to have linear distance using a~different proof idea?       
\end{problem}

We believe this open problem is very important. Indeed, if the conjecture were true, that would give us a~\emph{very explicit} construction of good qLDPC codes where we just require that the local code $\cL$ and its dual $\cL^\perp$ have large distances simultaneously. And since there are plenty of such codes (e.g., the Reed-Solomon codes), a~probabilistic construction and quite complicated proofs of robustness, given in~\cite{Panteleev&Kalachev:stoc2022} and in the current work, could be completely eliminated.

Note that recently Leverrier and Z{\'e}mor proposed~\cite{Leverrier:focs2022} a~very interesting simplification of our construction of good qLDPC codes from~\cite{Panteleev&Kalachev:stoc2022}. This new construction called \emph{quantum Tanner codes} has a~more natural interpretation and also gives codes with better lower distance bounds at the expense of increasing the weights of the parity-check matrices. We believe that the strong form of robustness, shown in the current work, can be used to simplify the proofs and to give better distance bounds in~\cite{Panteleev&Kalachev:stoc2022, Leverrier:focs2022}.

Though the qLDPC codes from~\cite{Panteleev&Kalachev:stoc2022, Leverrier:focs2022} have linear minimum distance, one also needs efficient decoding algorithms for such codes to use them in practice. It was conjectured in~\cite{Panteleev&Kalachev:stoc2022} that small-set-flip decoding algorithm from~\cite{Leverrier:2015} can be used to correct in linear time any adversarial errors up to the constant fraction of the code length. Quite recently~\cite{Gu:stoc2023:qpdpc-decoder, Leverrier:qldpcdecoder:2023} this conjecture was confirmed for both types of good qLDPC codes. Moreover, almost at the same time Dinur \emph{et al.}~\cite{Dinur:decoders} showed that the $3$-term complexes, equivalent to the complexes $\cA \otimes_G \cB$ from \cite[Remark on p.~379]{Panteleev&Kalachev:stoc2022} and similar to the ones naturally associated with the classical LTCs from~\cite{Dinur:stoc2022}, also give linear time decodable  good qLDPC codes if the product of the corresponding local codes is two-sided robustly testable. It is very important to note that in~\cite{Dinur:decoders} the authors also independently show (in the case of the binary field $\F_2$) the result equivalent to our Theorem~\ref{th:rand-expanding} on the optimal robustness, and also show the connection between the robustness and the coboundary expansion of product codes, similar to the one given in Appendix~\ref{app:prod-exp}.

Finally, let us mention that quite surprisingly qLDPC codes of large distances also found several interesting applications outside of the area of error correcting codes. Recently, they were used to prove new breakthrough results in complexity theory~\cite{Hopkins:focs2022, Anshu:stoc2023} and to provide spaces with exotic properties in systolic geometry~\cite{Freedman:2021}.

\section{Product-expansion}


In this section, we define a~property of a~collection of codes called \emph{product-expansion}\footnote{The \emph{product-expansion} property in terms of dual product codes $\cC_1\boxplus\cC_2$ was first defined in~\cite{Panteleev&Kalachev:stoc2022} in a~different form, which, as we will see in Appendix~\ref{app:prod-exp}, is equivalent (in the strongest case) to the current definition. Later, this property was reformulated in a~very elegant form by Leverrier and Z{\'e}mor in~\cite{Leverrier:focs2022}, where it is called \emph{robustness with resistance to puncturing}. Here we further simplify this definition and also consider its multi-dimensional version.} and show its connection to agreement testability of product of codes. Before we proceed, let us first remind some standard definitions and introduce some new notation related with the products of codes.

Given linear codes $\cC_1,\dots,\cC_m$ over $\F_q$ we can define the (\emph{tensor}) \emph{product code} 
\[
\cC_1\otimes\dots\otimes\cC_m := \{c\in \F_q^{n_1\times\dots\times n_m} \mid \forall i\in [m]\ \forall \ell\in \cL_i\colon c|_\ell\in  \cC_i\},
\]
where $\F_q^{n_1\times\dots\times n_m}$ is the set of functions $c\colon[n_1]\times \dots\times [n_m] \to \F_q$  and $\cL_i$ is the set of  lines parallel to the $i$-th axis in the $m$-dimensional grid $[n_1]\times \dots\times [n_m]$, i.e., 
\[
\cL_i := \{\{x + s\cdot e_i \mid s \in [n_i] \} \mid x\in [n_1]\times \dots\times [n_m], x_i = 0 \}.
\]
Here $e_i$ denotes the vector $(0,\dots,0,1,0\dots,0) \in [n_1]\times \dots \times [n_m]$ with $1$ at the $i$-th position. 

It is convenient to introduce a~notation for the dual code of a~product code~\cite{Wolf:1965, Chien:1973}. For linear codes $\cC_1\subseteq \F_q^{n_1}$, $\cC_2\subseteq \F_q^{n_2}$ we denote by $\cC_1\boxplus \cC_2$ the code $(\cC_1^\bot\otimes \cC_2^\bot)^\bot=\cC_1\otimes \F_q^{n_2}+\F_q^{n_1}\otimes \cC_2 \subseteq \F_q^{n_1\times n_2}$.
Given a~collection $\cC = (\cC_i)_{i\in [m]}$ of linear codes over~$\F_q$, we can define the codes 
\[
\cC^{(i)} := \F_q^{n_1} \otimes\dots\otimes \cC_i \otimes\dots\otimes \F_q^{n_m} = \{ c\in \F_q^{n_1\times\dots\times n_m} \mid \forall \ell\in \cL_i\colon c|_\ell\in  \cC_i \} .
\]
It is clear that $\cC_1\otimes \dots \otimes \cC_m = \cC^{(1)}\cap \dots \cap \cC^{(m)}$ and $\cC_1\boxplus \dots \boxplus \cC_m = \cC^{(1)}+ \dots + \cC^{(m)}$. Note that every code $\cC^{(i)}$ is the direct sum of $|\cL_i| = \frac{1}{n_i}\prod_{j\in [m]} n_j$ copies of the code $\cC_i$. For $x\in \F_q^{n_1\times\dots\times n_m}$ we denote by $|x|_i$ and $\norm{x}_i$ respectively the number and the fraction of the lines $\ell\in\cL_i$ such that $a|_\ell \ne 0$. It is clear that $\norm{x}_i = \frac{1}{|\cL_i|}|x|_i$. Let us also recall that by $|x|$ and $\norm{x}$ we denote respectively the \emph{Hamming weight} (i.e., the number of non-zero entries) and the \emph{normalized Hamming weight} (i.e., the fraction of non-zero entries) of $x$. Now we are ready to give our main definition.

\begin{definition}[product-expansion]
Given a~collection $\cC = (\cC_i)_{i\in [m]}$ of linear codes $\cC_i\subseteq \F_q^{n_i}$, we say that $\cC$ is \emph{$\rho$-product-expanding} if every codeword $c\in \cC_1\boxplus \dots \boxplus \cC_m$ can be represented as a~sum $c = \sum_{i\in[m]} a_i$ where $a_i\in \cC^{(i)}$ for all $i\in [m]$, and the following inequality holds:
\begin{equation}\label{eq:prod-exp}
\rho\sum_{i\in [m]} \norm{a_i}_i \le \norm{c} .    
\end{equation}
\end{definition}

It is not hard to check that (\ref{eq:prod-exp}) can be also expressed as
\[ \rho\sum_{i\in [m]} n_i\abs{a_i}_i \le \abs{c}.\]

\begin{remark}\label{rm:exam-non-robust}
Note that from the definition of $\rho$-product-expansion it immediately follows that any pair of codes $(\cC_1, \cC_2)$ defining a~CSS code (i.e., we have $H_1 H_2^* = 0$ for their parity-check matrices) is \emph{not} $\rho$-product-expanding. Indeed, since $x\in \cC_1\boxplus \cC_2$ \Iff $H_1 x H_2^* = 0$ this just follows from the fact that the identity matrix $(\delta_{ij})_{n\times n}\in \F_q^{n\times n}$ is a~codeword of $\cC_1 \boxplus\cC_2$, and $(\delta_{ij})_{n\times n}$ can not be obtained as a~sum of less than $n$ rows or columns.  Thus the pair $(\cC_1, \cC_2)$ is not $\rho$-product-expanding for any fixed $\rho$ as $n\to\infty$.
\end{remark}

We can also define the \emph{product-expansion factor} $\rho(\cC)$ for the collection $\cC$ as the maximal value of $\rho$ such that $\cC$ is $\rho$-product-expanding. In Appendix~\ref{sc:chain} it is shown that the product-expansion  corresponds to the coboundary expansion~\cite{Linial:2006,Gromov:2010} in the~tensor product complex obtained from the collection of codes $\cC = (\cC_i)_{i\in [m]}$. This tensor product complex can be best understood as a~\emph{local system} (also known as a~\emph{cellular sheaf}%
\footnote{Local systems are also often called \emph{cellular sheaves} (see Remark~\ref{rm:sheaves}).}%
) defined on the clique complex%
    \footnote{The \emph{clique complex} $\mathbf{X}(\cG)$ of a~graph $\cG$ is the simplicial complex with the set of vertices $V(\cG)$ where $S\in  X(\cG)$ \Iff the set $S$ gives a~\emph{clique} in $\cG$, i.e., the vertices from $S$ are pairwise connected in~$\cG$.} 
$\mathbf{X}(K_{n_1,\dots,n_m})$ of the complete $m$-partite graph $K_{n_1,\dots,n_m}$. With this interpretation in mind, the product-expansion factor $\rho(\cC)$ can be seen as a~natural generalization of the \emph{normalized Cheeger constant} of a~graph also known as the \emph{conductance}. Indeed, in the special case when $m=2$ and the codes $\cC_1$, $\cC_2$ are the repetition $[n, 1, n]$ code, $\rho(\cC_1,\cC_2)$ is exactly the conductance of the bipartite graph $K_{n,n}$.  

Note that there is a~degenerate case when $\cC_i=\F_q^{n_i}$ for some $i\in [m]$ in a~collection $\cC = (\cC_i)_{i\in [m]}$, in which case we call this collection \emph{degenerate}. In this case,  $\cC^{(i)}=\F_q^{n_1\times \cdots\times n_m}$, which implies $\cC_1\boxplus\cdots\boxplus \cC_m=\F_q^{n_1\times \cdots\times n_m}$, and therefore the collection $\cC$ is $(1/n_i)$-product-expanding. For example, if $n_i = 1$ (i.e., $\cC_i = \F_q$), then the collection is always $1$-product-expanding\footnote{This degenerate case appears only because we used here the \emph{normalized} Cheeger constant.} \emph{independently} of the codes $\cC_j$, $j\ne i$.

Note that $\rho$-product-expansion of a~non-degenerate collection of codes $\cC$ implies that $\delta(\cC_i)\ge \rho$ for all $i\in [m]$. Let us formulate this observation in a~more general form:
If a~non-degenerate collection $\cC=(\cC_i)_{i\in [m]}$ of codes $\cC_i\subsetneq \F_q^{n_i}$ is a~$\rho$-product-expanding, then each subcollection $\cC_I=(\cC_i)_{i\in I}$, $I\subseteq [m]$, is also $\rho$-product-expanding (Lemma \ref{lemma:subset-exp} in Appendix~\ref{sc:auxlemmas}).


Let us now show that in the case of two codes $\rho$-product-expansion is equivalent to $\rho$-agreement testability.
\begin{lemma}\label{lemma:prodexp-vs-agreement}
    A pair of codes $\cC_1,\cC_2\subseteq \F_q^n$ is $\rho$-product-expanding \Iff $\cC_1\otimes \cC_2$ is $\rho$-agreement testable.
\end{lemma}
\begin{proof}
    First, let us show that product-expansion implies agreement testability. Consider a~$\rho$-product-expanding pair of codes $(\cC_1,\cC_2)$ and some matrices $c_1\in\cC^{(1)}$, $c_2\in\cC^{(2)}$. 
    Since $c_1-c_2\in \cC_1\boxplus \cC_2$ and the pair $(\cC_1,\cC_2)$ is $\rho$-product-expanding we have $c_1-c_2=a_1+a_2$ for some $a_1\in\cC^{(1)}$, $a_2\in\cC^{(2)}$ such that 
    \begin{equation}\label{eqn:prodexp2agreement}
        \rho(\|a_1\|_1+\|a_2\|_2)\le \|c_1-c_2\|.
    \end{equation}
    Put $c:=c_1-a_1$. Then $c\in \cC^{(1)}$ and $c=c_2+a_2\in \cC^{(2)}$. Hence $c\in \cC^{(1)}\cap\cC^{(2)} = \cC_1\otimes \cC_2$ and \eqref{eqn:prodexp2agreement} is equivalent to~\eqref{eqn:agreement}. Therefore $\cC_1\otimes \cC_2$ is $\rho$-agreement testable.
    
    Now let us show that the agreement testability implies product expansion. Suppose $\cC_1\otimes \cC_2$ is $\rho$-agreement testable. Consider a word $c=c_1+c_2\in \cC_1\boxplus \cC_2$, $c_1\in \cC^{(1)}$, $c_2\in\cC^{(2)}$. Then from agreement testability there exists $c'\in \cC_1\otimes\cC_2$ such that $\rho(\|c_1-c'\|_1+\|c_2+c'\|_2)\le \|c_1+c_2\|=\|c\|$. Put $a_1:=c_1-c'$, $a_2:=c_2+c'$. Then $c=a_1+a_2$ and $\rho(\|a_1\|_1+\|a_2\|_2)\le \|c\|$, therefore the pair of codes $(\cC_1,\cC_2)$ is $\rho$-product-expanding.
\end{proof}

Denote by $\Gr(n,k)$ the~\emph{Grassmannian}, i.e., the set of all linear subspaces in $\F_q^n$ of dimension $k$.
In the next several sections, we will prove the following theorem. 

\begin{restatable}{theorem}{ThMain}\label{th:rand-expanding}
    For every $\eps_1\in(0,1)$, $\eps_2\in(0,1)$ there exists $\rho>0$ such that a~pair of codes $(\cC_1,\cC_2)$ picked uniformly at random from $\Gr(n,k_1)\times \Gr(n,k_2)$, where $k_i \le n_i(1-\eps_i)$, $i\in [2]$, is $\rho$-product-expanding with high probability as $n\to\infty$.
\end{restatable}

\begin{restatable}[two-sided robustness]{corollary}{CrMain}\label{cr:main}
    For every $R_1\in(0,1)$, $R_2\in(0,1)$ there exists $\rho>0$ such that for a~pair of codes $(\cC_1,\cC_2)$ picked uniformly at random from $\Gr(n,k_1) \times \Gr(n,k_2)$ the codes $C_1\otimes C_2$ and $C_1^\perp\otimes C_2^\perp$ are $\rho$-robustly testable with high probability as $n\to\infty$, $k_1/n\to R_1$, and $k_2/n\to R_2$.
\end{restatable}

We will often use the following short-hand notations: $x(I,\cdot) := x|_{I \times [n]}$, $x(\cdot,J) := x|_{[n]\times J}$, and $x(I,J) := x|_{I\times J}$, where $x\in \F_q^{n\times n}$, $I,J\subseteq [n]$. If $I\subseteq [n]$ we denote by $\F_q^I$ the linear space $\{x\in \F_q^n \mid \forall i\in [n]\setminus I\colon x_i = 0 \}$. Hence, by definition, we assume that $\F_q^I \subseteq \F_q^n$.

\subsection{Proof Outline}

Let us briefly explain an~idea of the proof. Consider two codes $\cC_1, \cC_2\subseteq\F_q^n$ of minimal distances%
\footnote{Here we refer to the standard minimal distance $d(\cC) := n\cdot \delta(\cC)$.}  $d_1=d(\cC_1)$ and $d_2=d(\cC_2)$, respectively. We say that $x\in \F_q^{n\times n}$ has a \emph{zero rectangle} $A\times B$, where $A,B\subseteq [n]$, if $|A|> n-d_1$, $|B|> n-d_2$, and $x(A,B)=0$. 
Informally, the proof involves the following two steps. 
\begin{enumerate}
\item If a~codeword $x\in \cC_1\boxplus \cC_2$ has a~zero rectangle $A\times B$, then $x$ is a~sum of $n-|A|$ rows from $\cC_2$ and $n-|B|$ columns from~$\cC_1$ (Lemma \ref{lemma:zeroangle});
\item A pair of random linear codes with high probability has the property that every codeword $x\in \cC_1\boxplus \cC_2$ of weight $\delta n^2$ has a zero rectangle of size $n(1-O(\delta))\times n(1-O(\delta))$;
\end{enumerate}

As the result, every codeword $x\in \cC_1\boxplus \cC_2$ of weight~$\delta n^2$ can be represented as a sum of $O(\delta)$ columns from~$\cC_1$ and rows from~$\cC_2$, which is exactly what is needed to prove the result.
The main difficulty on this path is to show the second statement from the above list. We need to prove that it holds for all codewords of weight $\le \delta_0 n^2$, where  $\delta_0>0$ is some fixed parameter, which does not depend on~$n$. Note that the number of such codewords does not exceed $q^{H_q(\delta_0)n^2}$. 

In fact, it is not hard to show that for some fixed $x$ the probability of the event $x\in \cC_1\boxplus \cC_2$ is bounded above as  $q^{-\Omega(n\rk x)}$ (Lemma \ref{lemma:prodexp-highrank}).
For the set of all matrices $x$ with $\rk x = \Omega(n)$ the probability that at least one of them is from the code $\cC_1\boxplus\cC_2$ can be estimated by the union bound. 
The main problem is the matrices where $\rk x=o(n)$ since the union bound is not enough here. 
We show that one can avoid this problem considering a~special property of the codes $\cC_1$ and $\cC_2$, which holds for random codes with high probability and excludes codewords in $\cC_1\boxplus\cC_2$ with a~small rank. 

To define this special property let us first introduce the notion of $\alpha$-sparseness. We say that a~vector $v\in \F_q^n$ is \emph{$\alpha$-sparse} if its Hamming weight is bounded above by $\alpha n$. A~subspace $V\subseteq \F_q^n$ is called \emph{$\alpha$-sparse} if it can be spanned by (zero or more) $\alpha$-sparse vectors. 
We say that a~subspace~$U\in \Gr(n,n-r)$ has property $(*)$ if for every $\alpha$-sparse subspace $V$ such that $\dim V\le r$ and $\alpha=H_q^{-1}(r/8n)$ we have $\dim (U\cap V)<\frac12\dim V$. By a~direct calculation we can show that almost all subspaces have property~$(*)$ (Lemma \ref{lemma:subspace-int}).

This special property is motivated by the fact that for a~matrix $x\in \cC_1\boxplus \cC_2$ with the column space $X:=\im x$ and the row space  $Y:=\im x^*$ we have that (Lemma \ref{lemma:rk-bound}):
\begin{equation}\label{eqn:rk-rcbound}
    \rk x\le \dim(\cC_1\cap X)+\dim(\cC_2\cap Y).
\end{equation} 
If codes $\cC_1$ and $\cC_2$ have property $(*)$, and we can guarantee that the spaces $X$ and $Y$ are \emph{$\alpha$-sparse}, then we have $\rk x=\dim X=\dim Y>n-\max(\dim \cC_1,\dim \cC_2)$. Indeed, otherwise by property $(*)$ we get $$\dim(\cC_1\cap X)+\dim(\cC_2\cap Y)<\frac12\dim X+\frac12\dim Y=\rk x,$$ which contradicts \eqref{eqn:rk-rcbound}.

Now consider codes $\cC_1$ and $\cC_2$ such that they have property $(*)$, and the dual product code $\cC_1\boxplus \cC_2$ does not have codewords of small weight and large rank. From the previous observations it follows that for a~sufficiently small  $\delta_0$ there are no codewords $x\in \cC_1\boxplus \cC_2$ of weight $\le \delta_0 n^2$ with $\alpha$-sparse row and column spaces.

Now it remains to show that for every codeword $x\in \cC_1\boxplus \cC_2$ of small weight $\le\delta n^2$, $\delta\le\delta_0$, either $x$ has a~zero rectangle of size $n(1-O(\delta))\times n(1-O(\delta))$ or there exists a~codeword $x'$ of weight $n^2 O(\delta)$ with $\alpha$-sparse row and column spaces. Indeed, for any matrix $x$ one can find a subset of rows and columns of weight $\ge \alpha n/2$. For a sufficiently small  $\delta_0\le\alpha^2/4$ the remaining rows (the index set~$A$) and columns (the index set~$B$) either form a zero rectangle or the weights of rows and columns in the submatrix $x(A, B)$ are bounded above by $\alpha n/2$, where $|A|,|B|\ge n(1-2\delta/\alpha)=n(1-O(\delta))$. In the last case, we can extend $x(A, B)$ to the codeword $x'\in \cC_1\boxplus \cC_2$, i.e., $x'(A,B') = x(A,B)$, in a~way that guarantees that all columns  
of $x'$ are spanned by the columns of $x'(\cdot,B)$, and all rows of $x'$ are spanned by the rows of $x'(A,\cdot)$ (Lemma \ref{lemma:extend-cwpart}), while the weight of all these rows and columns does not exceed $\alpha$. Therefore the row and the column spaces of $x'$ are both $\alpha$-sparse.

\subsection{Properties of Random Linear Subspaces}
Let us introduce some additional notations we will need in the following lemmas. For $n\in\NN$, $\alpha>0$ define the set $S(n,\alpha):=\{x\in \F_q^n\mid |x|\le\alpha n\}$. A~subspace of $\F_q^n$ is called \emph{$\alpha$-sparse} if it has a~basis consisting of vectors from $S(n,\alpha)$.
We will also need the~\emph{$q$-ary entropy function} $H_q:[0,1]\to [0,1]$ and its inverse $H_q^{-1}:[0,1]\to [0,1-1/q]$:
\[
    H_q(x) = x\log_q(q-1)-x\log_q x-(1-x)\log_q(1-x),
\]
where $H_q^{-1}(y)$ is the unique\footnote{Note that the function $H_q$ is monotonic and, hence, is invertible.} $x\in [0,1-1/q]$ such that $H_q(x)=y$.

In the following, we will also need the set $P(n,r_a,r_b):=\Gr(n,n-r_a)\times \Gr(n,n-r_b)$. It is known that $\abs{\Gr(n,k)} = \qbin{n}{k}$, where $\qbin{n}{k} := \prod_{i=0}^{k-1}\frac{q^{n-i}-1}{q^{k-i}-1}$ is the~\emph{$q$-binomial coefficient}.

Let us now state some known bounds~\cite[Lemma~4]{Koetter:2008} on the $q$-binomial coefficients $\qbin{n}{k}$. 
\begin{restatable}{lemma}{QbinBounds}\label{lemma:qbinbounds}
    For $q$-binomial coefficients we have the following bounds:
    $$q^{k(n-k)}\le \qbin{n}{k}\le 4q^{k(n-k)}.$$
\end{restatable}
\noindent A proof of this lemma is given in Appendix \ref{sc:proof-qbinbounds}.

Now, before we move to the next lemma, let us note that if $k\le m\le n$, then we have:
$$\frac{\binom{m}{k}_q}{\binom{n}{k}_q}=\prod_{i=0}^{k-1}\underbrace{\frac{q^m-q^i}{q^n-q^i}}_{\le q^{m-n}}\le q^{(m-n)k}.$$

\begin{lemma}\label{lemma:subsp-prob}
    For a~subspace $V\in \Gr(n, v)$ the probability that for random subspace $U\in \Gr(n,u)$ the dimension $\dim (U\cap V)\ge k$, is at most $4q^{-k(n+k-v-u)}$.
\end{lemma}
\begin{proof}
    The condition $\dim(U\cap V)\ge k$ means that there exist $k$-dimensional subspace $W\subseteq V$ such that $W\subseteq U$ which can be rewritten as $U^\perp\subseteq W^\perp$ where $U^\perp$ is the orthogonal subspace to $U$. Hence, the fraction of such subspaces $U$ is not more than
    \begin{equation*}
    \frac{\binom{\dim V}{\dim W}_q \binom{\dim W^\perp}{\dim U^\perp}_q}{\binom{n}{\dim U}_q}
    = \binom{v}{k}_q\frac{\binom{n-k}{n-u}_q}{\binom{n}{n-u}_q} \le 4q^{k(v-k)}\cdot q^{((n-k)-n)(n-u)}=4q^{-k(-v+k+n-u)}.\qedhere
    \end{equation*}
\end{proof}

Now, let us give an~upper bound on the probability that for random codes $\cC_1,\cC_2\subseteq\F_q^n$ the code $\cC_1\boxplus \cC_2$ has some fixed codeword $x$ of large rank.

\begin{lemma}\label{lemma:prodexp-highrank}
    For a~matrix $x\in\F_q^{n\times n}$ of rank $\rk x\ge \min(r_1,r_2)$ the probability that $x\in \cC_1\boxplus \cC_2$ for a~random pair of codes $(\cC_1,\cC_2)\in P(n,r_1,r_2)$ 
    is at most $5q^{-\frac14 r_1r_2}$ if $\min(r_1,r_2)\ge 2$.
\end{lemma}
\begin{proof}
    Let $h_i\in\F_q^{r_i\times n}$ be a parity check matrix of the~code $\cC_i$,  $i\in [2]$. Without loss of generality suppose that $r_1\le r_2$. 
    Consider $r':=\floorbr{r_1/2}$. Since $r_2\ge r_1\ge 2$, we have $r_1/4<r'\le r_1/2$, $r_1-r'\ge r_1/2$, $r_2-r'\ge r_2/2$. Let us estimate the number of codes $\cC_1\in\Gr(n, n-r_1)$ such that $\rk h_1x\le r'$.
    We have 
    \[r' \ge  \rk h_1x=\dim\im h_1x=\dim h_1\im x=\dim(\im x/\ker h_1|_{\im x})=\dim(\im x)-\dim(\underbrace{\ker h_1}_{\cC_1}\cap \im x), 
    \]
    i.e. $\dim (\im x\cap \cC_1)\ge \dim\im x-r'\ge r_1-r'\ge r_1/2$. By Lemma \ref{lemma:subsp-prob} we have
    \begin{multline*}
    \PP(\rk h_1x\le r')\le\PP(\dim (\im x\cap \cC_1)\ge \dim\im x-r')\\
    \le 4q^{-(\dim\im x-r')(n+(\dim\im x-r')-\dim\im x-(n-r_2))}=4q^{-(\dim\im x-r')(r_2-r')}\le 4q^{-\frac14 r_1r_2}. 
    \end{multline*}
    Now we fix a~code $\cC_1$ and a~matrix $h_1$ such that $\rk h_1x\ge r'$ and estimate the probability that $h_1x h_2^*=0$. Note that the condition $h_1xh_2^*=0$ is equivalent to $x\in \cC_1\boxplus \cC_2$, hence whether it holds or not depends only on the~codes $\cC_1$ and $\cC_2$ and it does not depend on a~particular choice of the parity check matrices $h_1$ and $h_2$.
    This can be rewritten as $h_2(h_1x)^*=0$, i.e., $\im (h_1x)^*\subseteq \ker h_2=\cC_2$. Let $U:=\im (h_1x)^*$, $u:=\dim U=\rk h_1x\ge r'$. We have
    \[\PP(U\subseteq \cC_2)=\PP(\cC_2^\bot\subseteq U^\bot)=\frac{\binom{n-u}{r_2}_q}{\binom{n}{r_2}_q}\le q^{-ur_2}\le q^{-r'r_2}.
    \]

    Note that $\PP(\rk h_1x > r'\wedge (h_1x)h_2^*=0)\le q^{-r'r_2}$ because we can first choose $\cC_1$, and when $\cC_1$ (and hence $U=\im (h_1x)^*$) is fixed, we independently choose $\cC_2$.
    Therefore, we can estimate
    \[
    \PP(h_1xh_2^*=0)\le \PP(\rk h_1x\le r')+\PP(\rk h_1x > r'\wedge (h_1x)h_2^*=0)\le 4q^{-\frac14 r_1r_2}+q^{-r'r_2}\le 5q^{-\frac14 r_1r_2}, 
    \]
    which completes the proof. \qedhere
\end{proof}

Let us recall that we say that a~subspace $V\subseteq \F_q^n$ is \emph{$\alpha$-sparse} if it can be spanned by vectors of Hamming weight at most $\alpha n$.

\begin{lemma}\label{lemma:subspace-int}
    For each $R\in (0,1)$, $n\in \NN$, $r\ge R n$ 
    a~random $(n-r)$-dimensional space $U$ with the~probability at least $1-4\frac{q^{-r/8}}{1-q^{-r/8}}$ for all $m\in [1, r]$ has the following property:
    \begin{itemize}
    \item[$(*)$] for each $\alpha$-sparse $m$-dimensional subspace $V\subseteq \F_q^n$ we have $\dim (U\cap V) < m/2$. 
    \end{itemize}
    where $\alpha=H_q^{-1}(R/8)$.
\end{lemma}
\begin{proof}

    
    Let us fix $m\in[1,r]$ and consider an~$\alpha$-sparse $m$-dimensional subspace $V\subseteq \F_q^n$. 
    By Lemma~\ref{lemma:subsp-prob} the probability that $V$ is a~counterexample for the property $(*)$ for $(n-r)$-dimensional subspace $U$ (i.e. $\dim(U\cap V)\ge m/2$), can be estimated as follows:
    \[
    \PP(\dim(U\cap V)\ge \ceilbr{m/2})\le 4q^{-\ceilbr{m/2}(n+\ceilbr{m/2}-m-(n-r))}\le 4q^{-m(r-m/2)/2}\le 4q^{-mr/4}.
    \]
    Note that there are at most $|S(n,\alpha)|^m$ ways to choose $m$ basis vectors for an~$\alpha$-sparse subspace $V$. We can estimate this number as follows:
    $$|S(n,\alpha)|^m\le \rbr{\sum_{i=0}^{\floorbr{n\alpha}}\binom{n}{i}(q-1)^i}^m\le q^{mn H_q(\alpha)}=q^{mnR/8}\le q^{mr/8}.
    $$
    Now we can use the union bound for all possible counterexamples of dimension $m$ to estimate the probability that for a~subspace $U$ property $(*)$ does not hold for this dimension:
    \begin{align*}
    \PP((*)\mbox{ does not hold for $\dim V=m$})&\le \sum_{V}\overbrace{\PP((*)\mbox{ does not hold with counterexample }V)}^{\le 4q^{-mr/4}}\\
    &\le q^{mr/8}\cdot 4q^{-mr/4}\le 4q^{-mr/8}.
    \end{align*}
    Finally, summing by all possible dimensions $m=\overline{1,r}$ we have
    \[
    \PP((*)\mbox{ does not hold for $\dim V\in [r]$})\le 4\sum_{m=1}^rq^{-mr/8}< 4\frac{q^{-r/8}}{1-q^{-r/8}}.\qedhere
    \]
\end{proof}

Note that if for an~$(n-r)$ dimensional subspace $U\subseteq \F_q^n$ we have property $(*)$ and $r\ge 2$, then the distance of the code $U$ is greater than $2\alpha n$.

\subsection{Proof of the Main Result}

Let us start from proving some auxiliary lemmas we need to obtain the main result. Recall that a~$k$-dimensional subspace $\cC\subseteq \F_q^n$ is called \emph{$[n,k,d]_q$-code} if $d:=n\cdot \delta(\cC)$, where the parameter $d=d(\cC)$ is called the \emph{minimal distance} of $\cC$, and the \emph{rate} of $\cC$ is defined as $R(\cC) := k/n$.

\begin{restatable}{lemma}{LemIntersection}\label{lemma:intersection}
    For linear codes $\cC_1,\cC_2,X,Y\subseteq \F_q^n$ we have: 
    $$(X\otimes Y)\cap (\cC_1\boxplus \cC_2)=(X\cap \cC_1)\otimes Y+X\otimes (Y\cap \cC_2).$$
\end{restatable}
\noindent A proof of this lemma is given in Appendix \ref{sc:proof-intersection}.

The following lemma is a~natural analog for dual product codes $\cC_1 \boxplus \cC_2$ of the well-known ``rectangle method'' initially developed for product codes $\cC_1\otimes \cC_2$ (see~\cite{Dinur:2006, Meir:2012}). 

\begin{lemma}[on zero rectangles]\label{lemma:zeroangle}
    Consider linear codes $\cC_1,\cC_2\subseteq \F_q^n$. If $x\in \cC_1\boxplus \cC_2$, and for $A_1,A_2\subseteq [n]$ such that $n-|A_i|<d(\cC_i)$, $i\in [2]$, we have $x(A_1,A_2)=0$, then $x$ can be represented as a~sum of $n-|A_1|$ rows from $\cC_2$ and $n-|A_2|$ columns from~$\cC_1$.
\end{lemma}

\begin{proof}
    \begin{figure}[ht]
        \centering
        \begin{tikzpicture}
                \filldraw[black!10!white] (4,1) rectangle (0,0) rectangle (1,4);
                \draw (0,0) rectangle (4,4);
                \draw[dashed] (0,1)--(4,1);
                \draw[dashed] (1,0)--(1,4);
                \node[] at (2.5,2.5) {0};
                \draw [decorate,decoration={brace,amplitude=8pt}]
                    (0,1) -- (0,4) node [midway,left,xshift=-0.8em] {$A_1$};
                \draw [decorate,decoration={brace,amplitude=8pt}]
                    (0,0) -- (0,1) node [midway,left,xshift=-0.8em] {$<d(\cC_1)$};
                \draw [decorate,decoration={brace,amplitude=8pt}]
                    (4,0) -- (1,0) node [midway,below,yshift=-0.8em] {$A_2$};
                \draw [decorate,decoration={brace,amplitude=6pt}]
                    (1,0) -- (0,0) node [midway,below,yshift=-0.8em] {$< d(\cC_2)$};
    
                \draw[->] (4.5,2)--node[above]{$-\ \Delta_1-\Delta_2$}(7,2);
                \begin{scope}[xshift=8cm]
                \filldraw[black!10!white] (1.5,1) rectangle (0,0) rectangle (1,1.5);
                \draw (0,0) rectangle (4,4);
                    \draw[dashed] (0,1)--(4,1);
                    \draw[dashed] (1,0)--(1,4);
                    \node[] at (2.5,2.5) {0};
                    \node[red] at (3,0.5) {0};
                    \node[red] at (0.5,3) {0};
                    \draw [decorate,decoration={brace,amplitude=8pt}]
                        (4,0) -- (1.5,0) node [midway,below,yshift=-0.8em] {$I_2$};
                \draw [decorate,decoration={brace,amplitude=8pt}]
                    (0,1.5) -- (0,4) node [midway,left,xshift=-0.8em] {$I_1$};
                \end{scope}
        \end{tikzpicture}
        \caption{Idea of the proof.}
        \label{fig:zeroangle}
    \end{figure}
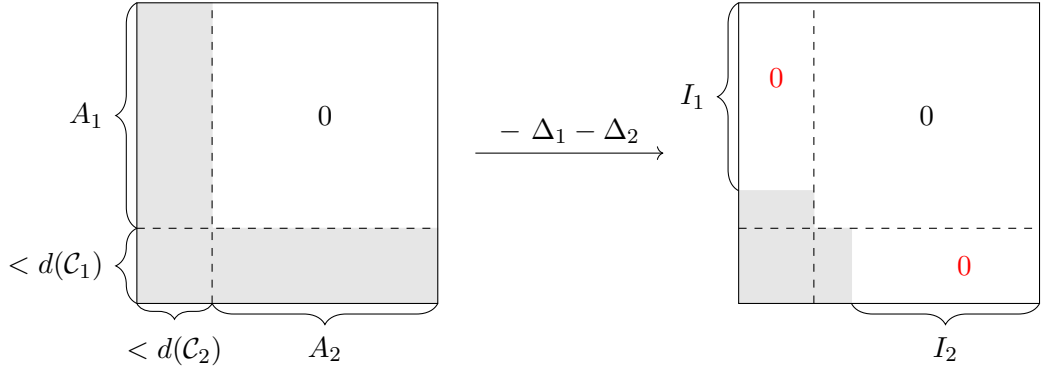
    Since $n-|A_i|<d(\cC_i)$, the index set $A_i$ contains an~information set%
    \footnote{
        An~\emph{information set} for a~linear code $\cC\subseteq\F_q^n$ is a smallest by inclusion index set $I\subseteq [n]$ such that for every $c\in\cC$ if $c|_I = 0$ then $c = 0$.  It is clear that for every $S\subseteq [n]$ such that $\abs{S} > n - d(\cC)$ if for some codeword $c\in \cC$ we have $c|_S = 0$ then $c=0$. Hence there should exist an~information set $I\subseteq S$.
    } 
    $I_i$ of the code $\cC_i$, $i\in [2]$. Let $\bar I_i:=[n]\setminus I_i$, $\bar A_i:=[n]\setminus A_i$. Therefore we can uniquely choose $\Delta_1\in \cC_1\otimes \F_q ^{\bar A_2}$ and $\Delta_2\in\F_q^{\bar A_1}\otimes \cC_2$ such that $\Delta_1(I_1,\bar A_2)=x(I_1,\bar A_2)$ and $\Delta_2(\bar A_1, I_2)=x(\bar A_1, I_2)$. Consider $x'=x-\Delta_1-\Delta_2$. Then we have $x'(I_1,\cdot)=0$, $x'(\cdot,I_2)=0$, and therefore $x'\in \F_q^{\bar I_1}\otimes \F_q^{\bar I_2}$ (the right part in Fig.~\ref{fig:zeroangle}). Since $x'\in \cC_1\boxplus \cC_2$, by Lemma~\ref{lemma:intersection} we get  
    \[
    x'\in \underbrace{(\F_q^{\bar I_1}\cap \cC_1)}_{=\{0\}}\otimes \F_q^{\bar I_2}+\F_q^{\bar I_1}\otimes \underbrace{(\F_q^{\bar I_2}\cap \cC_2)}_{=\{0\}}=\{0\}.
    \]
    
    Here we have $\F_q^{\bar I_i}\cap \cC_i=\{0\}$ since $I_i$ is an~information set $\cC_i$. Therefore $x'=0$ and $x=\Delta_1+\Delta_2$, and moreover $\Delta_1$ is a~sum of  no more than $|\bar A_2|=n-|A_2|$ non-zero columns from~$\cC_1$, and  $\Delta_2$ is a sum of no more than $|\bar A_1|=n-|A_1|$ non-zero rows from $\cC_2$. This completes the proof.
\end{proof}

\begin{lemma}\label{lemma:rk-bound}
    Consider linear codes $\cC_1,\cC_2\subseteq\F_q^n$, a~codeword~$x\in \cC_1\boxplus \cC_2$, the spaces $X=\im x$ and $Y=\im x^*$ spanned respectively by the columns and rows from~$x$. Then we have 
    \[\rk x\le \dim (X\cap \cC_1)+\dim (Y\cap \cC_2).\]
\end{lemma}
\begin{proof}
    It is not hard to check that $x\in (X\otimes Y)\cap (\cC_1\boxplus \cC_2)$. Now, using Lemma~\ref{lemma:intersection}, we get $x\in (X\cap \cC_1)\otimes Y+X\otimes (Y\cap \cC_2)$. Therefore 
    $x=x_1+x_2$ for some $x_1\in (X\cap \cC_1)\otimes Y$, $x_2\in X\otimes (Y\cap \cC_2)$.  
    Now we see that for the matrices $x, x_1, x_2$ we have $\rk x_1\le \dim (X\cap \cC_1)$, $\rk x_2\le \dim(Y\cap \cC_2)$, and therefore
    \begin{equation*}
        \rk x\le \rk x_1+\rk x_2\le \dim (X\cap \cC_1)+\dim (Y\cap \cC_2).\qedhere
    \end{equation*}
\end{proof}
\begin{lemma}\label{lemma:extend-cwpart}
Consider linear codes $\cC_1$, $\cC_2\subseteq \F_q^n$, and a~codeword~$x\in \cC_1\boxplus \cC_2$. Then for every $A_1,A_2\subseteq [n]$ there exists $x'\in \cC_1\boxplus \cC_2$ such that\footnote{In other words, all the rows (resp. columns) from~$x'$ are linear combinations of the rows from $x'(A_1,\cdot)$  (resp. columns from~$x(\cdot,A_2)$.} $x'(A_1, A_2)=x(A_1,A_2)$ and $\rk x'=\rk x(A_1,A_2)$.
\end{lemma}
\begin{proof}
    First, we define projection operators $\pi_i\colon \F_q^n\to \F_q^{A_i}$, $\pi_i x=x|_{A_i}$, $i\in [2]$. Consider $\cC_i|_{A_i}:=\pi_i\cC_i$. Let us fix a~basis $B^0_i$ of the subspace~$\cC_i|_{A_i}$ and extend it to the basis $B_i$ of the space~$\F_q^{A_i}$.
    Now we define an~operator $g_i$ on the basis $B_i$. Since for every vector $e\in B^0_i$ there exists a~vector~$\bar e\in \cC_i$ such that $\pi \bar e=e$, we can put by definition that $g_i e:=\bar e$. Since for every vector  $e\in B_i\setminus B^0_i$ there exists a~vector~$\bar e\in \F_q^n$ such that $\pi \bar e=e$, we can put by definition that $g_i e:=\bar e$. Now from the above definition we have $\pi_i g_i=\id$ and $g_i (\cC_i|_{A_i})\subseteq \cC_i$.
    
    If we put $x':=(g_1\otimes g_2) x(A_1,A_2)$, then we obtain that  
    $$x'(A_1,A_2)=(\pi_1\otimes \pi_2)x'=(\pi_1\otimes \pi_2)(g_1\otimes g_2)x(A_1,A_2)=x(A_1,A_2).$$ 
    Note that $x(A_1,A_2)=(\pi_1\otimes \pi_2)x\in \cC_1|_{A_1}\otimes \F_q^{A_2}+\F_q^{A_1}\otimes \cC_2|_{A_2}$.
    Besides that we have 
    $$(g_1\otimes g_2)(\cC_1|_{A_1}\otimes \F_q^{A_2}+\F_q^{A_1}\otimes \cC_2|_{A_2})=g_1 \cC_1|_{A_1}\otimes \im g_2+\im g_1\otimes g_2 \cC_2|_{A_2}\subseteq \cC_1\otimes \F_q^n+\F_q^n\otimes \cC_2,
    $$
    and therefore $x'\in \cC_1\boxplus \cC_2$.
    It is also not hard to see that $\rk x'\le\rk x(A_1,A_2)$. Finally, since $x'(A_1,A_2)=x(A_1,A_2)$, we get $\rk x'=\rk x(A_1,A_2)$.
\end{proof}

\begin{lemma}\label{lemma:no-sparse-lowrank}
    Consider linear codes $\cC_1\in \Gr(n,n-r_1)$, $\cC_2\in\Gr(n, n-r_2)$ that satisfy property~$(*)$. Then $d(\cC_i)> \alpha_i n$, and  every codeword $x\in \cC_1\boxplus \cC_2$ with weight $|x|\le \alpha_1\alpha_2 n^2/4$ and rank $\rk x\le\min(r_1,r_2)$ has a zero rectangle $A\times B$, where $\alpha_i:=H_q^{-1}(\frac{r_i}{8n})$, $|A|\ge n-\frac{|x|}{\alpha_2 n/2}$, $|B|\ge n-\frac{|x|}{\alpha_1 n/2}$.
\end{lemma}
\begin{proof}
    First, let us fix a~vector $v\in\F_q^n$, $|v|\le\alpha_i n$, define a 1-dimensional space $V:=\abr{v}\subseteq \F_q^n$, and use property~$(*)$ for $\cC_i$, $i\in [2]$. As a~result, we obtain $\dim(\cC_i\cap V)<1/2$, and hence $v\ne \cC_i$. Therefore, we get $d(\cC_i)> \alpha_i n$  ($i\in [2]$).

    Now consider a codeword $x\in \cC_1\boxplus \cC_2$ with weight $|x|<\alpha_1\alpha_2n^2/4$ and rank $\rk x\le\min(r_1,r_2)$. Let $A$ (resp.~$B$) be the set of indexes of rows (resp. columns) of the matrix $x$ of weight no more than $\alpha_2 n/2$ (resp., $\alpha_1n/2$), $\bar A:=[n]\setminus A$, $\bar B:=[n]\setminus B$. Then $|x|\ge \max(|\bar B|\alpha_1,|\bar A|\alpha_2) n/2$, and therefore  
    $$|\bar A|\le \frac{|x|}{\alpha_2 n/2}\le \alpha_1n/2<d(\cC_1),\qquad |\bar B|\le \frac{|x|}{\alpha_1 n/2}\le \alpha_2n/2<d(\cC_2)$$
    This means that $A$ (resp., $B$) contains an information set of the code $\cC_1$ (resp., the code $\cC_2$).
    Now assume that $A\times B$ is not a zero rectangle of the codeword~$x$, i.e. $X(A,B)\ne 0$. 
    By Lemma~\ref{lemma:extend-cwpart} there exists a codeword~$x'\in \cC_1\boxplus \cC_2$ such that $x'(A,B)=x(A,B)$, all the rows (resp. columns) of $x'$ are spanned by the rows of $x'(A,\cdot)$ (resp. columns of $x'(\cdot,B)$). 
    Since $|x'(i,\cdot)|\le |x'(i,B)|+|\bar B|=|x(i,B)|+|\bar B|\le \alpha_2 n$ for $i\in A$, the space of rows $Y$ of the matrix $x'$ is $\alpha_2$-sparse. 
    In the same way, the space of columns $X$ of $x'$ is also $\alpha_1$-sparse. From property~$(*)$ for the codes $\cC_1$ and $\cC_2$ we get $\dim (\cC_1\cap X)<\frac12\dim X$, $\dim (\cC_2\cap Y)<\frac12\dim Y$, and therefore by Lemma~\ref{lemma:rk-bound} we get a~contradiction:
    $$\rk x'\le \dim(\cC_1\cap X)+\dim(\cC_2\cap Y) < \frac12\dim X+\frac12\dim Y=\rk x',$$
    and $X(A,B)=0$. Hence $A\times B$ is a zero rectangle in $x$ with the required parameters, and the proof is complete.
\end{proof}

\ThMain*
\begin{proof}
    Let $\alpha_i=H_q^{-1}(\eps_i/8)$ ($i\in [2]$),  
    \begin{equation}\label{eqn:rho}
        \rho=\frac12\min\!\rbr{\frac{\alpha_1\alpha_2}{4}, H_q^{-1}\!\!\rbr{\frac{\eps_1\eps_2}{8}}}.
    \end{equation}

    Consider $r_1:=n-k_1$, $r_2=n-k_2$. We have $r_1\ge \eps_1 n$, $r_2\ge \eps_2 n$. 
    For simplicity we assume that $n$ is sufficiently large and $\min(r_1,r_2)\ge 8$.
    Consider two cases of ``bad'' pairs of codes and estimate the probability.
    \begin{enumerate}
        \item The set $P_1\subseteq P(n,r_1,r_2)$ of all pairs $(\cC_1,\cC_2)$ such that there exists $x\in \cC_1\boxplus \cC_2$  of weight $|x|\le 2\rho n^2$ and of rank $\rk x\ge\min(r_1,r_2)$. Using Lemma~\ref{lemma:prodexp-highrank} and applying the union bound for all possible words of weight $\le 2\rho n^2$ we get the bound:
        \begin{multline}\label{eqn:P1}
        \PP(P_1)
        \le \sum_{i=0}^{\floorbr{2\rho n^2}}\binom{n^2}{i}(q-1)^i\cdot 5q^{-\frac14 r_1r_2}
        \le q^{n^2 H_q(2\rho)}\cdot 5q^{-\frac14 n^2 \eps_1\eps_2} \\
        \le 5q^{\frac18 n^2 \eps_1\eps_2-\frac14 n^2 \eps_1\eps_2}=5q^{-\frac18 n^2 \eps_1\eps_2}.
        \end{multline}
        \item The set  $P_2\subseteq P(n,r_1,r_2)$ of all such pairs $(\cC_1,\cC_2)$ such that one of the codes does not have property~$(*)$. By Lemma~\ref{lemma:subspace-int} the probability that $\cC_i$ does not have property~$(*)$ is not more than $4\frac{q^{-r_i/8}}{1-q^{-r_i/8}}$. Hence we get:
        \begin{equation}\label{eqn:P2}
          \PP(P_2)\le 4\frac{q^{-r_1/8}}{1-q^{-r_1/8}}+4\frac{q^{-r_2/8}}{1-q^{-r_2/8}}\le 16 q^{-\frac18\min(r_1,r_2)}\le 16q^{-\frac18 n\min(\eps_1,\eps_2)}. 
        \end{equation}
    \end{enumerate}
    Combining \eqref{eqn:P1} and \eqref{eqn:P2} as $n\to\infty$ we get $\PP(P(n,r_1,r_2)\setminus P_1\setminus P_2)\to 1$.
    
    It remains to show for an arbitrary pair $(\cC_1,\cC_2)\in P(n,r_1,r_2)\setminus P_1\setminus P_2$ that it is $\rho$-product-expanding. Let us fix a~non-zero codeword $x\in \cC_1\boxplus \cC_2$. Now consider two cases:
    \begin{enumerate}
        \item Case $|x|\ge 2\rho n^2$. In this case, $x$ can always be represented as a sum of columns from $\cC_1$ and rows from $\cC_2$. However the total number of rows and columns is $2n\le\frac{|x|}{\rho n}$.
        \item Case $|x|< 2\rho n^2$. Since $(\cC_1,\cC_2)\not\in P_1$, then $\rk x<\min(r_1,r_2)$. Hence since  $|x|\le \frac{\alpha_1\alpha_2}{4}n^2$ and the codes $\cC_1,\cC_2$ have property $(*)$, then by Lemma~\ref{lemma:no-sparse-lowrank} the word $x$ has a~zero rectangle $A_1\times A_2$, where $|A_1|\ge n-\frac{2|x|}{\alpha_2 n}$, $|A_2|\ge n-\frac{2|x|}{\alpha_1 n}$, and $d(\cC_i)\ge \alpha_i n$. Then by Lemma~\ref{lemma:zeroangle} the word $x$ can be represented as a sum of  $n-|A_2|$ columns from  $\cC_1$ and $n-|A_1|$ rows from $\cC_2$.
        Therefore we get the bound $(n-|A_1|)+(n-|A_2|)\le \frac{2|x|}{\alpha_1 n}+\frac{2|x|}{\alpha_2 n}\le \frac{|x|}{(\alpha_1\alpha_2/4)n}<\frac{|x|}{\rho n}$.
    \end{enumerate}
    Hence in all cases we get that $x$ can be represented as no more than $\frac{|x|}{\rho n}$ columns from  $\cC_1$ and rows from $\cC_2$. Therefore a pair of codes $(\cC_1,\cC_2)$ is $\rho$-product-expanding.
\end{proof}

\CrMain*
\begin{proof}

    Fix $\delta=\min(R_1,R_2,1-R_1,1-R_2)/2$.
    Applying Theorem \ref{th:rand-expanding} with $\eps_i=(1-R_i-\delta)$, $i\in [2]$ we have that for some $\rho_1>0$ a random pair of codes $(\cC_1,\cC_2)\in \Gr(n,k_1)\times \Gr(n,k_2)$ is $\rho_1$-product expanding with high probability when $k_i\le (R_i+\delta)n$, $i\in [2]$.
    Applying Theorem \ref{th:rand-expanding} with $\eps_i=(R_i-\delta)$, $i\in [2]$ we have that for some $\rho_2>0$ a random pair of codes $(\cC_1^\bot,\cC_2^\bot)\in \Gr(n,n-k_1)\times \Gr(n,n-k_2)$ is $\rho_2$-product expanding with high probability when $k_i\ge (R_i-\delta)n$, $i\in [2]$.
    
    Since $k_i/n\to R_i$, when $n$ is big enough, we have $n(R_i-\delta)\le k_i\le (R_i+\delta)n$, 
    $i\in [2]$. Hence, for a random pair of codes $(\cC_1,\cC_2)\in \Gr(n,k_1)\times \Gr(n,k_2)$ both pairs $(\cC_1,\cC_2)$ and $(\cC_1^\bot,\cC_2^\bot)$ are $\rho'$-product-expanding with high probability as $n\to\infty$, 
    where $\rho'=\min(\rho_1,\rho_2)$.
    

    By Lemma \ref{lemma:prodexp-vs-agreement}, $\rho'$-product expansion of pairs $(\cC_1,\cC_2)$ and $(\cC_1^\bot,\cC_2^\bot)$ is equivalent to $\rho'$-agreement testability of product codes $\cC_1\otimes\cC_2$ and $\cC_1^\bot\otimes \cC_2^\bot$, which by \cite[Lemma 2.9]{Dinur:stoc2022} implies $\rho$-robust testability of these codes with $\rho=\frac{\rho'}{2(\rho'+1)}$.
\end{proof}

\section{Discussions and Open Problems}

In this paper, we showed that a~random pair of linear codes is product-expanding with high probability, which also implies that they are robustly and agreement testable. We conjecture that a~collection of more than two random codes is also product-expanding with high probability. However the proof method we used here can not be directly applied to this more general case. We think that the product-expansion in the case of more than two codes can be used to construct high-dimensional analogs of the two-dimensional chain complexes from~\cite{Dinur:stoc2022, Panteleev&Kalachev:stoc2022}.  

Note that the one-dimensional version of $\rho$-product-expansion (i.e., for only \emph{one} code~$\cC$) naturally corresponds to the relative minimum distance of $\cC$. If for a~collection of codes $\cC = (\cC_i)_{i\in[m]}$ we denote by $\rho(\cC)$ its \emph{expansion factor}, i.e. \[
\rho(\cC) := \sup \{\rho \in [0,1] \mid \cC \text{ is  $\rho$-product-expanding}\},
\]   
then in the one-dimensional case we have $\rho(\cC) = \delta(\cC)$. Thus, in the general case of $m$ codes, we can view $\rho(\cC_1,\dots,\cC_m)$ as an~$m$-dimensional analog of the relative minimum distance.  Note that in the two-dimensional case $\rho(\cC_1,\cC_2)$ is the maximal $\rho$ such that $\cC_1\otimes\cC_2$ is $\rho$-agreement testable. 

It is natural to study the optimal expansion factor $\rho(\cC_1,\dots,\cC_m)$ of a collection of $m$ codes of rates $R_1,\dots,R_m$.     
In this case, an~analogue of theorem~\ref{th:rand-expanding} is the Gilbert-Varshamov bound, which says that a~random code of rate $1-\eps$ with high probability has distance not less than $H^{-1}_q(\eps)$. This bound, in an~asymptotic regime when $\eps\to 0$, $n\to\infty$, matches (up to constant factors) the Hammning upper bound, and therefore the relative distance of the best codes is $\Theta(H^{-1}_q(\eps))=\Theta(\eps/\log \eps)$.

It is interesting to generalize the bound on $\rho$ for the multi-dimensional case (i.e., for the case of several codes $\cC_1,\dots,\cC_m$). Theorem~\ref{th:rand-expanding} gives a~lower bound on~\eqref{eqn:rho} in the two-dimensional case. One can obtain an~upper bound from the fact (Lemma~\ref{lemma:subset-exp}) that the relative minimal distance%
\footnote{Recall that in one-dimensional case we have $\rho(\cC) = \delta(\cC)$.}
of codes $\cC_1$ and $\cC_2$ is not less than $\rho$ and from the upper Hamming bound. Hence, using these observations we get $\rho=O(\min(H_q^{-1}(\eps_1),H_q^{-1}(\eps_2)))$ as $\min(\eps_1,\eps_2)\to 0$, where $1-\eps_i$ is the rate of $\cC_i$. Let us stress that the value of the parameter $\rho$ in \eqref{eqn:rho} can be estimated as
\[
\rho = O\Bigl(H_q^{-1}\!\rbr{\textstyle\frac{\eps_1}{8}}H_q^{-1}\!\rbr{\textstyle\frac{\eps_2}{8}}\Bigr)=O\!\rbr{\frac{\eps_1\eps_2}{\log_q \eps_1\log_q \eps_2}}\quad\mbox{ as }\quad \eps_1\to 0, \eps_2\to 0. 
\]
Now we see that if both rates tend to  $1$ (the case $\eps_1\to 0, \eps_2\to 0$), then the lower bound is much worse than the upper one. However, in the case that is used in the recent constructions of good qLDPC codes~\cite{Panteleev&Kalachev:stoc2022,Leverrier:focs2022} we have $\eps_1+\eps_2=1$. Therefore assuming that $\eps_1=\eps, \eps_2=1-\eps$ the value of $\rho$ in \eqref{eqn:rho} coincides with an upper bound (up to a~constant factor): 
\[
\rho=\Theta\rbr{\min\rbr{H_q^{-1}\!\rbr{\textstyle\frac{\eps}{8}}H_q^{-1}\!\rbr{\textstyle\frac{1-\eps}{8}},H_q^{-1}\!\rbr{\textstyle\frac{\eps(1-\eps)}{8}}}}=\Theta\!\rbr{H_q^{-1}(\eps)}\quad\mbox{ as }\quad \eps\to 0. 
\]

Note that it is also possible to get a~bound on the parameter $\rho$ that is independent of the finite field size. This bound is analogous to the Singleton bound $d\le n-k + 1$ for an $[n,k,d]_q$-code, which can be also formulated as 
\[\delta \le \eps + 1/n\]
where $\delta := d/n$ and $\eps := 1 - k/n$.

For a~pair of codes $\cC_1\in\Gr(n,k_1)$ and $\cC_2\in\Gr(n,k_2)$ it is possible to prove a~similar bound:
\[
\rho(\cC_1,\cC_2)\le\eps_1\eps_2+1/n,
\]
where $\eps_i:=1-k_i/n$ (see Proposition~\ref{prop:upper} in Appendix~\ref{app:upper}).
It is interesting whether this bound is asymptotically tight. More precisely, whether for arbitrary $\eps_1,\eps_2\in (0,1)$ there exists an infinite family  $(\cC_1^i,\cC_2^i)_{i\in\NN}$ of code pairs with rates  $R(\cC_1^i) \ge 1-\eps_1$ and  $R(\cC_2^i) \ge 1-\eps_2$ such that $\rho(\cC_1^i,\cC_2^i)\ge \eps_1\eps_2$, $i\in\NN$.

\section*{Acknowledgment}

This work was supported by the Ministry of Science and Higher Education of the Russian Federation (Grant 075-15-2020-801).

\bibliographystyle{plain}
\bibliography{codes.bib}

\begin{thebibliography}{10}

\bibitem{Anshu:stoc2023}
Anurag Anshu, Nikolas~P. Breuckmann, and Chinmay Nirkhe.
\newblock {NLTS} hamiltonians from good quantum codes.
\newblock In {\em Proceedings of the 55th Annual ACM Symposium on Theory of
  Computing}, STOC 2023, page 1090–1096, New York, NY, USA, 2023. Association
  for Computing Machinery.

\bibitem{Ben-Sasson:2006}
Eli {Ben-Sasson} and Madhu Sudan.
\newblock Robust locally testable codes and products of codes.
\newblock {\em Random Structures \& Algorithms}, 28(4):387--402, 2006.

\bibitem{Breuckmann:balanced:2021}
Nikolas~P. Breuckmann and Jens~N. Eberhardt.
\newblock Balanced product quantum codes.
\newblock {\em IEEE Transactions on Information Theory}, 67(10):6653--6674,
  October 2021.

\bibitem{Breuckmann:2021}
Nikolas~P. Breuckmann and Jens~Niklas Eberhardt.
\newblock Quantum low-density parity-check codes.
\newblock {\em PRX Quantum}, 2(4):040101, October 2021.

\bibitem{Brown:1982}
Kenneth~S. Brown.
\newblock {\em Cohomology of {{Groups}}}, volume~87 of {\em Graduate {{Texts}}
  in {{Mathematics}}}.
\newblock {Springer}, {New York, NY}, 1982.

\bibitem{CSS:1996}
A.~R. Calderbank and Peter~W. Shor.
\newblock Good quantum error-correcting codes exist.
\newblock {\em Phys. Rev. A}, 54:1098--1105, Aug 1996.

\bibitem{Chien:1973}
R.~Chien and S.~Ng.
\newblock Dual product codes for correction of multiple low-density burst
  errors.
\newblock {\em IEEE Transactions on Information Theory}, 19(5):672--677,
  September 1973.

\bibitem{Curry:2014}
Justin~Michael Curry.
\newblock {\em Sheaves, cosheaves and applications}.
\newblock PhD thesis, University of Pennsylvania, 2014.

\bibitem{Dinur:stoc2022}
Irit Dinur, Shai Evra, Ron Livne, Alexander Lubotzky, and Shahar Mozes.
\newblock Locally testable codes with constant rate, distance, and locality.
\newblock In {\em Proceedings of the 54th Annual ACM SIGACT Symposium on Theory
  of Computing}, STOC 2022, pages 357--374, New York, NY, USA, June 2022.
  Association for Computing Machinery.

\bibitem{Dinur:decoders}
Irit Dinur, Min-Hsiu Hsieh, Ting-Chun Lin, and Thomas Vidick.
\newblock Good quantum ldpc codes with linear time decoders.
\newblock In {\em Proceedings of the 55th Annual ACM Symposium on Theory of
  Computing}, STOC 2023, page 905–918, New York, NY, USA, 2023. Association
  for Computing Machinery.

\bibitem{Dinur:2017}
Irit Dinur and Tali Kaufman.
\newblock High dimensional expanders imply agreement expanders.
\newblock In {\em 2017 IEEE 58th Annual Symposium on Foundations of Computer
  Science (FOCS)}, pages 974--985, October 2017.

\bibitem{Dinur:2006}
Irit Dinur, Madhu Sudan, and Avi Wigderson.
\newblock Robust local testability of tensor products of {LDPC} codes.
\newblock In Josep D{\'i}az, Klaus Jansen, Jos{\'e} D.~P. Rolim, and Uri Zwick,
  editors, {\em Approximation, Randomization, and Combinatorial Optimization.
  Algorithms and Techniques}, Lecture Notes in Computer Science, pages
  304--315, Berlin, Heidelberg, 2006. Springer.

\bibitem{Dinur:2022ex}
Irit Dinur and Thomas Vidick.
\newblock The tensor of two codes that are dual to each other is not robustly
  testable.
\newblock In preparation.

\bibitem{firstGoodQueryLocally2023}
Uriya~A. First and Tali Kaufman.
\newblock On {Good} 2-{Query Locally Testable Codes} from {Sheaves} on {High
  Dimensional Expanders}, May 2023.

\bibitem{Freedman:2021}
Michael Freedman and Matthew Hastings.
\newblock Building manifolds from quantum codes.
\newblock {\em Geometric and Functional Analysis}, June 2021.

\bibitem{gallager1963}
Robert~G. Gallager.
\newblock {\em Low-density parity-check codes}.
\newblock M.I.T. Press, Cambridge, MA, 1963.

\bibitem{ghrist2014elementary}
Robert~W. Ghrist.
\newblock {\em Elementary applied topology}.
\newblock CreateSpace Independent Publishing Platform, 2014.

\bibitem{Goldreich:2012}
Oded Goldreich and Or~Meir.
\newblock The tensor product of two good codes is not necessarily robustly
  testable.
\newblock {\em Information Processing Letters}, 112(8):351--355, April 2012.

\bibitem{LTC:2006}
Oded Goldreich and Madhu Sudan.
\newblock Locally testable codes and pcps of almost-linear length.
\newblock {\em Journal of the ACM}, 53(4):558--655, July 2006.

\bibitem{Gromov:2010}
Mikhail Gromov.
\newblock Singularities, expanders and topology of maps. part 2: from
  combinatorics to topology via algebraic isoperimetry.
\newblock {\em Geometric and Functional Analysis}, 20(2):416--526, August 2010.

\bibitem{Gu:stoc2023:qpdpc-decoder}
Shouzhen Gu, Christopher~A. Pattison, and Eugene Tang.
\newblock An efficient decoder for a linear distance quantum ldpc code.
\newblock In {\em Proceedings of the 55th Annual ACM Symposium on Theory of
  Computing}, STOC 2023, page 919–932, New York, NY, USA, 2023. Association
  for Computing Machinery.

\bibitem{Hansen:2019}
Jakob Hansen and Robert Ghrist.
\newblock Toward a spectral theory of cellular sheaves.
\newblock {\em Journal of Applied and Computational Topology}, 3(4):315--358,
  December 2019.

\bibitem{Hastings:2021:fiber}
Matthew~B. Hastings, Jeongwan Haah, and Ryan O'Donnell.
\newblock Fiber bundle codes: breaking the {$N^{1/2}
  \operatorname{polylog}(N)$} barrier for quantum {LDPC} codes.
\newblock In {\em Proceedings of the 53rd Annual ACM SIGACT Symposium on Theory
  of Computing}, pages 1276--1288. Association for Computing Machinery, New
  York, NY, USA, June 2021.

\bibitem{Hopkins:focs2022}
Max Hopkins and Ting-Chun Lin.
\newblock Explicit lower bounds against {$\Omega(n)$}-rounds of sum-of-squares.
\newblock In {\em 2022 IEEE 63rd Annual Symposium on Foundations of Computer
  Science (FOCS)}, pages 662--673, 2022.

\bibitem{Hu:2020}
Chuan-Shen Hu.
\newblock A brief note for sheaf structures on posets, October 2020.

\bibitem{Kalachev:2023}
Gleb Kalachev.
\newblock High-dimensional expansion of product codes is stronger than robust
  and agreement testability, August 2023.

\bibitem{Kashiwara:1984}
Masaki Kashiwara.
\newblock The riemann-hilbert problem for holonomic systems.
\newblock {\em Publications of the Research Institute for Mathematical
  Sciences}, 20(2):319--365, 1984.

\bibitem{Koetter:2008}
Ralf Koetter and Frank~R. Kschischang.
\newblock Coding for errors and erasures in random network coding.
\newblock {\em IEEE Transactions on Information Theory}, 54(8):3579--3591,
  August 2008.

\bibitem{Leverrier:2015}
Anthony Leverrier, Jean-Pierre Tillich, and Gilles Z{\'e}mor.
\newblock Quantum expander codes.
\newblock In {\em 2015 IEEE 56th Annual Symposium on Foundations of Computer
  Science}, pages 810--824, October 2015.

\bibitem{Leverrier:focs2022}
Anthony Leverrier and Gilles Z{\'e}mor.
\newblock Quantum tanner codes.
\newblock In {\em 2022 IEEE 63rd Annual Symposium on Foundations of Computer
  Science (FOCS)}, pages 872--883, Los Alamitos, CA, USA, nov 2022. IEEE
  Computer Society.

\bibitem{Leverrier:qldpcdecoder:2023}
Anthony Leverrier and Gilles Zémor.
\newblock {\em Efficient decoding up to a constant fraction of the code length
  for asymptotically good quantum codes}, pages 1216--1244.
\newblock 2023.

\bibitem{Linial:2006}
Nathan Linial and Roy Meshulam.
\newblock Homological connectivity of random 2-complexes.
\newblock {\em Combinatorica}, 26(4):475--487, August 2006.

\bibitem{Mackay:2004}
D.~J.~C. MacKay, G.~Mitchison, and P.~L. McFadden.
\newblock Sparse-graph codes for quantum error correction.
\newblock {\em IEEE Transactions on Information Theory}, 50(10):2315--2330,
  October 2004.

\bibitem{Meir:2012}
Or~Meir.
\newblock On the rectangle method in proofs of robustness of tensor products.
\newblock {\em Information Processing Letters}, 112(6):257--260, March 2012.

\bibitem{Panteleev&Kalachev:2021}
Pavel Panteleev and Gleb Kalachev.
\newblock Quantum {LDPC} codes with almost linear minimum distance.
\newblock {\em IEEE Transactions on Information Theory}, pages 1--1, 2021.

\bibitem{Panteleev&Kalachev:stoc2022}
Pavel Panteleev and Gleb Kalachev.
\newblock Asymptotically good quantum and locally testable classical {LDPC}
  codes.
\newblock In {\em Proceedings of the 54th Annual ACM SIGACT Symposium on Theory
  of Computing}, STOC 2022, pages 375--388, New York, NY, USA, June 2022.
  Association for Computing Machinery.

\bibitem{Polishchuk:1994}
Alexander Polishchuk and Daniel~A. Spielman.
\newblock Nearly-linear size holographic proofs.
\newblock In {\em Proceedings of the Twenty-Sixth Annual {{ACM}} Symposium on
  {{Theory}} of {{Computing}}}, {{STOC}} '94, pages 194--203, {New York, NY,
  USA}, May 1994. {Association for Computing Machinery}.

\bibitem{shepard1985cellular}
A.D. Shepard.
\newblock {\em A cellular description of the derived category of a stratified
  space}.
\newblock PhD thesis, Brown University, 1985.

\bibitem{CSS2:1996}
A.~M. Steane.
\newblock Error correcting codes in quantum theory.
\newblock {\em Phys. Rev. Lett.}, 77:793--797, Jul 1996.

\bibitem{Steenrod:1943}
N.~E. Steenrod.
\newblock Homology with local coefficients.
\newblock {\em Annals of Mathematics}, 44(4):610--627, 1943.

\bibitem{Tillich&Zemor:2009}
J.~{Tillich} and G.~{Z\'emor}.
\newblock Quantum {LDPC} codes with positive rate and minimum distance
  proportional to $n^{1/2}$.
\newblock In {\em 2009 IEEE International Symposium on Information Theory},
  pages 799--803, June 2009.

\bibitem{Valiant:2005}
Paul Valiant.
\newblock The tensor product of two codes is not necessarily robustly testable.
\newblock In Chandra Chekuri, Klaus Jansen, Jos{\'e} D.~P. Rolim, and Luca
  Trevisan, editors, {\em Approximation, Randomization and Combinatorial
  Optimization. Algorithms and Techniques}, Lecture Notes in Computer Science,
  pages 472--481, Berlin, Heidelberg, 2005. Springer.

\bibitem{Wolf:1965}
J.~Wolf.
\newblock On codes derivable from the tensor product of check matrices.
\newblock {\em IEEE Transactions on Information Theory}, 11(2):281--284, April
  1965.

\bibitem{Zeeman:1962}
E.~C. Zeeman.
\newblock Dihomology: I. relations between homology theories.
\newblock {\em Proceedings of the London Mathematical Society}, 3(1):609--638,
  1962.

\end{thebibliography}

\appendix

\section{Auxiliary Lemmas}\label{sc:auxlemmas}

\begin{lemma}\label{lemma:subset-exp}
Let $\cC=(\cC_i)_{i\in [m]}$ be a~$\rho$-product-expanding collection of codes $\cC_i\subsetneq \F_q^{n_i}$. Then each subcollection $\cC_I=(\cC_i)_{i\in I}$, $I\subseteq [m]$, is also $\rho$-product-expanding.
\end{lemma}
\begin{proof}
    Without loss of generality we can assume that $I=[m']$, $1\le m'<m$.
    Since $\cC_i\neq \F_q^{n_i}$ we can find 
    $j_i\in [n_i]$ such that vector $e^i_{j_i}:=(0,...,0,1,0,...,0)\in \F_q^{n_i}$ with 1 on the~$j_i$-th position is not a~codeword of $\cC_i$ for $m'+1\le i\le m$. We can also find $b_i\in \cC_i^\bot$ such that $\langle b_i,e^i_{j_i}\rangle=1$ for $m'+1\le i\le m$.
    For a~word $c\in \boxplus_{i\in [m']}\cC_i$ define the word $c':=c\otimes \bigotimes_{i=m'+1}^{m}e^i_{j_i}\in \cC_1\boxplus\cdots\boxplus\cC_m$.
    Since $\cC$ is $\rho$-product-expanding we have $c=\sum_{i\in[m]}a_i$, $a_i\in \cC^{(i)}$ such that $\rho\sum_{i\in [m]}\|a\|_i\le \|c\|$.
    Define linear maps $\phi_i:\F_q^{n_i}\to \F_q$, $\phi_i x:=\langle x,b_i\rangle$ for $i=\overline{m'+1,m}$, 
    $$\phi:=\id_1\otimes \cdots\otimes \id_{m'}\otimes \phi_{m'+1}\otimes\cdots\otimes \phi_m:\F_q^{n_1\times \cdots\times n_m}\to \F_q^{n_1\times \cdots\times n_{m'}}.$$ Let $a'_i := \phi a_i$ for $i\in[m']$. By construction, we have $\|a'_i\|_i\le N\|a_i\|_i$ for $i\in [m']$, $\phi a_i=0$ for $i\ge m'+1$ and $\phi c'=c$ where $N:=\prod_{i=m'+1}^m n_{i}$ is a normalization coefficient. Hence 
    \[
    c=\sum_{i\in[m']}a'_i\qquad\mbox{and}\qquad
    \rho\sum_{i\in [m']}\|a'_i\|_i\le \rho \sum_{i\in [m]}N\|a_i\|_i\le  N\|c'\|=\|c\|.
    \qedhere
    \]
\end{proof}
\subsection{Proof of Lemma \ref{lemma:qbinbounds}}\label{sc:proof-qbinbounds}
Let us recall the statement of the lemma.
\QbinBounds*
\begin{proof} Indeed, we have:
$$
    \qbin{n}{k} = \prod_{i=0}^{k-1}\frac{q^{n-i}-1}{q^{k-i}-1} 
    = q^{k(n-k)}\prod_{i=0}^{k-1}\frac{1-q^{i-n}}{1-q^{i-k}}.
$$
Since $k\le n$ we have $1-q^{i-n}\ge 1-q^{i-k}$. Therefore $\prod_{i=0}^{k-1}\frac{1-q^{i-n}}{1-q^{i-k}}\ge 1$. On the other hand, we see that
$$ 
\prod_{i=0}^{k-1}(1-q^{i-k})=\prod_{i=1}^{k}(1-q^{-i})\ge (1-q^{-1})\Bigl(1-\sum_{i=2}^\infty q^{-i}\Bigr)=1-q^{-1}-q^{-2}\ge \frac{1}{4},
$$
where the last inequality holds since $q\ge 2$. Thus we have
$$\prod_{i=0}^{k-1}\frac{1-q^{i-n}}{1-q^{i-k}}\le \prod_{i=0}^{k-1}\frac{1}{1-q^{i-k}}\le 4,$$
which completes the proof.
\end{proof}

\subsection{Proof of Lemma \ref{lemma:intersection}}\label{sc:proof-intersection}
Let us recall the statement of the lemma.
\LemIntersection*
We will use that tensor product of spaces is distributive over sums and intersections: For any spaces $A,B,C\subseteq\F_q^n$ we have
\[
(A+B)\otimes C=(A\otimes C)+(B\otimes C),\qquad (A\cap B)\otimes C=(A\otimes C)\cap (B\otimes C).
\]
In the proof we will many times apply the following lemma.
\begin{lemma}\label{lemma:linsp-eq}
    For any linear codes $A,B,C,D,E,F\subseteq \F_q^n$ we have: 
    $$A\otimes B\cap (C\otimes D+E\otimes F)=A\otimes B\cap (C'\otimes D+E\otimes F),
    $$
    where $C' = C\cap(A+E)$.
\end{lemma}
\begin{proof}
    The inclusion of the right-hand-side to the left-hand-side is obvious. Thus it remains to check the inclusion in the opposite direction. 
    Let $x\in A\otimes B\cap (C\otimes D+E\otimes F)$. Then $x=x_1+x_2$, where $x_1\in C\otimes D$, $x_2\in E\otimes F$. Therefore we get $x_1=x-x_2\in E\otimes F+A\otimes B\subseteq (E+A)\otimes (F+B)$. However, since $x_1\in C\otimes D$, we obtain that $x_1\in (C\otimes D)\cap ((E+A)\otimes (F+B))\subseteq (C\cap (A+E))\otimes D$. Therefore we get
    $x=x_1+x_2\in ((C\cap (A+E))\otimes D+E\otimes F)\cap A\otimes B$, and the proof is complete.
\end{proof}

\begin{proof}[Proof of Lemma \ref{lemma:intersection}]
    The proof follows from the following sequence of equations: 
    \begin{align*}
        (X\otimes Y)\cap (\cC_1\boxplus \cC_2)&=(X\otimes Y)\cap (\cC_1\otimes \F_q^n+\F_q^n\otimes \cC_2)\\
        (\mbox{applying two times Lemma~\ref{lemma:linsp-eq}})\quad &=(X\otimes Y)\cap (\cC_1\otimes \underline{(Y+\cC_2)}+\underline{(X+\cC_1)}\otimes \cC_2)\\
        (\mbox{using }\cC_1\otimes \cC_2\subseteq (X+\cC_1)\otimes \cC_2)\quad &=(X\otimes Y)\cap (\cC_1\otimes Y+(X+\cC_1)\otimes \cC_2)\\
        (\mbox{applying Lemma~\ref{lemma:linsp-eq}})\quad &=(X\otimes Y)\cap (\cC_1\otimes Y+(X+\cC_1)\otimes \underline{(\cC_2\cap Y)})\\
        (\mbox{using }\cC_1\otimes (\cC_2\cap Y)\subseteq \cC_1\otimes Y)\quad &=(X\otimes Y)\cap (\cC_1\otimes Y+X\otimes (\cC_2\cap Y))\\
        (\mbox{applying Lemma~\ref{lemma:linsp-eq}})\quad &=(X\otimes Y)\cap (\underline{(\cC_1\cap X)}\otimes Y+X\otimes (\cC_2\cap Y))\\
        (\mbox{every term in a sum is from }X\otimes Y)\quad &= (\cC_1\cap X)\otimes Y+X\otimes (\cC_2\cap Y).
    \end{align*}
    The lemma is proved. \qedhere
\end{proof}

\section{Coboundary Expansion for Product Codes}\label{sc:chain}

In this subsection, we show that our definition of product-expansion given before is just an~instance of coboundary expansion~\cite{Linial:2006,Gromov:2010} of some natural $m$-dimensional cochain complex associated with the codes $\cC_1,\dots,\cC_m$. To define this complex, we also need some additional notations and definitions from~\cite{Panteleev&Kalachev:stoc2022} related with the local systems of coefficients defined on a~poset~$X$, which we usually view as the cell poset of some cell complex. It was suggested in~\cite{Panteleev&Kalachev:stoc2022} to extend the theory of high-dimensional expanders to this much more general context by replacing the standard Hamming norm with the block Hamming norm, measuring the number of non-zero blocks of data assigned to the elements of the poset $X$. Note that similar ideas were also independently developed in~\cite{firstGoodQueryLocally2023} in the special case of simplicial complexes.

Let us briefly recall that a~\emph{chain complex} over $\F_q$ is a~vector space $\sC = \bigoplus_{i\in\ZZ}\sC_i$ (over $\F_q$) with some fixed linear map $\partial\colon\sC
\rightarrow \sC $ called the~{\tmem{boundary map}} such that $\partial^2 = 0$ and $\partial
\sC_{i} \subseteq \sC_{i-1}$ for all $i\in \mathbb{Z}$. A~complex $(\sC,\partial)$ is usually represented by a~sequence 
\[\cdots\xrightarrow{\partial_{i+1}} \sC_{i} \xrightarrow{\partial_{i}} \sC_{i-1} \xrightarrow{\partial_{i-1}}\cdots\]
of vector spaces $\sC_i$ and maps\footnote{Sometimes when it does not cause confusion we omit the indexes in the boundary maps $\partial_i$ and simply write $\partial$ instead of $\partial_i$.} $\partial_i := \partial
|_{\sC_i}\colon\sC_i \rightarrow \sC_{i - 1}$ such that $\partial_{i}\circ \partial_{i+1} = 0$ for all $i\in\Z$. The term $\sC_i$ in a~complex~$\sC$ is called the space of \emph{$i$-chains} and the assertion $\partial_{i}\circ \partial_{i+1} = 0$ is equivalent to $\im \partial_{i+1} \subseteq \ker \partial_i$, which allows us to consider for every $i\in\Z$ the quotient vector space $H_i(\sC) = \ker \partial_i / \im \partial_{i+1}$ called the~\emph{$i$-th homology group} of the complex~$\sC$. 
The elements from $Z_i  := \ker \partial_i$ and $B_i  := \im \partial_{i+1}$ are called
the {\tmem{$i$-cycles}} and {\tmem{$i$-boundaries}} of~$\sC$, respectively. 

In the context of error correcting codes, we are interested in complexes $\sC$ over $\F_q$ with $\tau$ non-zero terms (\emph{$\tau$-term} complexes), where each term~$\sC_i$ is \emph{based}, i.e.,  comes with a~distinguished basis ${X_i \subseteq \sC_i}$ over $\F_q$, which elements are called \emph{$i$-cells}. In many cases, it is convenient to consider $i$-chains $c\in \sC_i$ as formal $\F_q$-linear combinations 
\[c = \sum_{x\in X(i)} c_x x\] 
of $i$-cells and identify $\sC_i$ with~$\F_q^{n_i}$, where $n_i := \abs{X(i)}$. Note that any space of formal linear combinations $\F_q S \cong \F_q^{|S|}$ is equipped with the standard \emph{inner product} $\langle a, b \rangle := \sum_{s\in S} a_s b_s$ and the \emph{Hamming norm} $\abs{a} := \abs{\supp a}$, where $\supp a := \{s \in S \mid a_s \ne 0\}$. We usually interpret terms in such chain complexes either as the~space of code symbols or the space of parity-checks. 

A~\emph{cochain complex} is the dual notion to the chain complex, i.e., we reverse all arrows and replace the maps by their transposed maps. In this case, it is common to add prefix ``co'' to all standard terms, i.e., instead of chains, cycles, etc., we consider \emph{co}chains, \emph{co}cycles, etc.. Hence we have the spaces $\sC^i$, $Z^i$, and $B^i$ of $i$-cochains, $i$-cocycles, and
$i$-coboundaries, respectively.

A~\emph{local system} of coefficients $\cF$ can be used to obtain a~(co)chain complex over $\F_q$ out of a~poset~$X$ by assigning to each element $x$ from $X$ a~vector space $\cF_x$ over $\F_q$ and defining for every pair $x \le x'$ from $X$ an~$\F_q$-linear map $\cF_{x\to x'}\colon \cF_x \to \cF_{x'}$  such that:  $\cF_{x\to x} = \id_{\cF_x}$ for all $x$ from $X$, and $\cF_{x'\to x''}\circ\cF_{x\to x'} = \cF_{x\to x''}$ for all triples $x \le  x' \le x''$ from $X$.  

\begin{remark}[for experts]\label{rm:sheaves}
Note that the general idea of the homology theory with local coefficients was originally proposed by Steenrod in the~1940s~\cite{Steenrod:1943}, while the idea of assigning abelian groups $\cF_x$ to the cells of an~arbitrary cell complex $X$ (simplicial, cubical, etc.) appeared later in the~1960s~\cite[p.~624]{Zeeman:1962}, where it was also connected to sheaf theory. Since that time it has been reappeared several times in different contexts~\cite{Kashiwara:1984, shepard1985cellular}, and nowadays such local systems are often called \emph{cellular sheaves}~\cite{Curry:2014, Hansen:2019}. They have found many interesting applications in computer science~\cite[Chapter~9]{ghrist2014elementary}. The connection of cellular sheaves to the classical ones on topological spaces is done by (1) defining the Alexandrov topology on a~poset $X$ where the open sets are the unions of the upper sets $U_x := \{y\in X \colon x \le y\}$; (2) assigning to each $U_x$ the abelian group $\cF_x$; and (3) extending this definition to arbitrary open sets~\cite[Section~4.2]{Curry:2014}, \cite[Theorem~6.1]{Hu:2020}. 
\end{remark}

In the current paper, we are interested in the case when the poset $X$ consists of all  axis-parallel hyperplanes in the $m$-dimensional grid $[n_1]\times\dots\times [n_m]$. Denote by $\widehat{[n]}$ the poset defined on the set $\{*\}\cup \{1,\dots,n\}$ where $x < y$ \Iff $x=*$ and $y\in [n]$. Now, we define $X$ as the direct product%
    \footnote{The \emph{direct product} of posets $(X_1,\le_1),\dots,(X_m,\le_i)$ is the set $X_1\times\dots\times X_m$ with the partial order $\le $ defined by the rule: $(x_1,\dots,x_m) \le_{X} (x'_1,\dots,x'_m)$ \Iff  $\forall i\in[m]$: $x_i \le_{i} x'_i$.} 
$\widehat{[n_1]}\times\dots\times \widehat{[n_m]}$. The poset $X$ is equipped with the following grading function: For every $x = (x_i)_{i\in [m]}\in X$ we put  $\dim x := \sum_{i\in [m]} \dim x_i$, where $\dim x_i := 0$ if $x=*$ and $\dim x_i := 1$ otherwise. This allows us to divide $X$ into  $m+1$ levels $X = X(0)\sqcup X(1) \sqcup \dots \sqcup X(m)$, where $X(i) := \{x\in X \mid \dim x = i \}$.
In the two-dimensional case (i.e., $m=2$), one can identify $X$ with different parts of an $n_1\times n_2$ matrix $M$: $X(2)$ is the index set of the $n_1 n_2$ individual elements $M(i,j)$; $X(1)$ is the index set for  $n_1$ rows $M(i,\cdot)$ and  $n_2$ columns $M(\cdot, j)$; and $X(0)$ is the index set for the one element representing the entire matrix $M(\cdot,\cdot)$. 

Note that such posets naturally arise as local neighborhoods of the vertices in the Cartesian product%
    \footnote{We assume that a~graph $\cG = (V,E)$ is represented by a~two-level graded poset $V\sqcup E$, where $v < e$ \Iff the vertex $v$ is incident to the edge $e$, and the Cartesian product of graphs is represented by the direct product of their posets.} 
$\cG_1\times \dots \times \cG_m$ of $m$ simple $n_i$-regular graphs and other cubical complexes having the same structure \emph{locally} (e.g., the complexes from~\cite{Panteleev&Kalachev:stoc2022} obtained as finite coverings of graph products). Indeed, one can represent each local neighborhood in the $n_i$-regular graph $\cG_i$ as the poset 
$\widehat{[n_i]}$, where $*$ denotes the vertex, while the elements of $[n_i]$ represent the $n_i$ connected to this vertex edges. Hence the local neighborhood of a~given vertex $v$ in $P = \cG_1\times \dots \times \cG_m$, defined as the subposet $\st_v P := \{x \in X \mid v \le x\}$, is clearly isomorphic%
\footnote{Recall that two graded posets are \emph{isomorphic} if there exists a~one-to-one map respecting the order and the grading.} 
to the graded poset $X = \widehat{[n_1]}\times\dots\times \widehat{[n_m]}$. Moreover, if we decrease the dimension of each element in $P_v$ by one, then the obtained graded poset can be identified with a~link of $v$ in $P$, which geometrically can be viewed as the intersection of a~sufficiently small $m$-dimensional sphere around $v$ with the geometrical realization of $P$. It is not hard to check that this link is just the clique complex $\mathbf{X}(K_{n_1,\dots,n_m})$ of the complete $m$-partite graph $K_{n_1,\dots,n_m}$. Indeed, every tuple $(x_1,\dots,x_m)\in X$ corresponds to the simplex 
\[\{x_i \mid x_i\ne * , i \in [m]\} \subseteq [n_1]\sqcup \dots \sqcup [n_m]\] 
from $\mathbf{X}(K_{n_1,\dots,n_m})$. For example, in the mentioned above two-dimensional case, the link is the complete bipartite graph $K_{n_1,n_2}$ (viewed as a~$1$-dimensional simplicial complex), where the edges correspond to the elements of the matrix, the vertices to the rows and columns, and the empty set~$\varnothing$ to the entire matrix.

Let $X$ be some finite poset, which we are going to use as an~index set. If a~vector space~$\sC$ is the~direct sum $\bigoplus_{x\in X} \cF_x$ of a~collection of vector spaces $(\cF_x)_{x\in X}$, then we can consider the elements of~$\sC$ as formal sums $\sum_{x\in X} a_x x$ of elements from $X$, where for every $x\in X$ the coefficient $a_x$ is from the vector space~$\cF_x$ called the \emph{local coefficient space} of~$x$ or the \emph{stalk} of $\cF$ at $x$ when we view $\cF$ as the cellular sheaf. In such cases, we also denote the vector space~$\sC$ by~$\cF X$ or by~$A X$ when all the local coefficient spaces are equal to the same space~$A$. If each local coefficient space $\cF_x$ comes with a~distinguished basis $\tilde{\cF_x}$, then we assume that the distinguished basis for $\cF X$ is the set $\{ax \mid a\in \tilde{\cF_x},\ x\in X \}$, in which case we say that $\cF X$ is \emph{based}.

Consider a~based chain complex $\sC=\cF X$ over $\F_q$. Let $a = \sum_{x \in X}a_x x \in \sC$, where each coefficient $a_x$ is from the based vector space $\cF_x$ over $\F_q$. We denote by $\wt(a)$ the standard Hamming weight of $a$, considered as a~vector over $\F_q$. We also consider the \emph{block weight}  $\wt_{X}(a)$ defined as the number of non-zero blocks in $a$, viewed as a~block vector $(a_{x})_{x\in X}$, i.e., $\wt_X (a) := |\supp_X a|$, where $\supp_X a := \{x\in X \mid a_x \ne 0\}$.

\begin{figure}
    \centering
    \begin{tikzpicture}
        \begin{scope}[xscale=1]
        \filldraw[red!30] (4.75,1) rectangle (5.25,2);
        \filldraw[blue!20] (1,4.75) rectangle (2,5.25);
        \filldraw[red!30] (4.75,4) rectangle (5.25,6);
        \filldraw[blue!20] (4,4.75) rectangle (6,5.25);
        \filldraw[red!40!blue!20] (4.75,4.75) rectangle (5.25,5.25);
        \draw[step=1] (1,1) grid (2,2);
        \draw[step=0.25,gray,dotted] (4,1) grid (5.99,1.99);
        \draw[step=1] (3.99,1) grid (6,2);
        \draw[step=1,shift={(0.25,0)}] (3.99,1) grid (5.75,2);
        \draw[step=1,shift={(0.5,0)}] (3.99,1) grid (5.5,2);
        \draw[step=1,shift={(0.75,0)}] (3.99,1) grid (5.25,2);
        \draw[step=0.25,gray,dotted] (1,4) grid (1.99,5.99);
        \draw[step=1] (1,3.99) grid (2,6);
        \draw[step=1,shift={(0,0.25)}] (1,3.99) grid (2,5.75);
        \draw[step=1,shift={(0,0.5)}] (1,3.99) grid (2,5.5);
        \draw[step=1,shift={(0,0.75)}] (1,3.99) grid (2,5.25);
        \draw[step=0.25] (3.99,3.99) grid (6,6);
        \draw[step=0.25,gray,dotted] (1,1) grid (1.99,1.99);
        \draw[->] (1.5,2.3) --node[right]{$g_1\otimes \id$} (1.5,3.7);
        \draw[->] (5,2.3) --node[right]{$g_1\otimes \id$} (5,3.7);
        \draw[->] (2.3,1.5) --node[above]{$\id\otimes g_2$} (3.7,1.5);
        \draw[->] (2.3,5) --node[above]{$\id\otimes g_2$} (3.7,5);
        \node[below] at (1.5,1) {$\F_q^{k_1\times k_2}$};
        \node[left] at (1,5) {$\F_q^{k_2}$};
        \node[below] at (5,1) {$\F_q^{k_1}$};
        \end{scope}
    \end{tikzpicture}
    \caption{Local system for $\sC(g_1, g_2)$.}
    \label{fig:local-syst}\label{fg:coboundary}
\end{figure}

To every linear code $\cC\subseteq \F_q^n$ defined by the encoding map $g\colon \F_q^k\to \F_q^n$ (i.e., $\im g = \cC$, $\rk g = k$) we can assign the complex 
\begin{equation}\label{eq:chain-g}
    \sC(g) := (\F_q^k\{*\} \xrightarrow{g} \F_q[n]),
\end{equation}
where we assume the straightforward isomorphisms $\F_q^k\{*\} \cong \F_q^k$ and $\F_q[n] \cong \F_q^n$. We can describe $\sC(g)$ by the local system $\cF$ on the poset $\widehat{[n]}$ defined as 
\[
\cF_x := \begin{cases} 
      \F_q^k & x = *; \\
      \F_q & x = i \in [n], \\
   \end{cases}
\]
and $\cF_{*\to i} := x \mapsto  (gx)|_i$, where $|_i$ denotes the projection on the $i$-th coordinate. 

Given a collection of codes $\cC_1,\dots,\cC_m$ defined by the corresponding encoding maps $g_1,\dots,g_m$ we can consider the tensor product complex $\sC(g_1,\dots,g_m) := \sC(g_1) \otimes \dots \otimes \sC(g_m)$. This tensor product complex can also be described by a~local system on the poset $X = \widehat{[n_1]}\times\dots\times \widehat{[n_m]}$ defined above. It is not hard to check that the code $\cC_1 \boxplus \dots \boxplus \cC_m$ corresponds to the space of $m$-coboundaries $B^{m}$ of the complex $\sC(g_1,\dots,g_m)$. Moreover, the collection $(\cC_i)_{i\in[m]}$ is $\rho$-product-expanding \Iff $\sC(g_1,\dots,g_m)$ is a coboundary expander in the following sense:
For every $c\in B^{m}$ we have 
\begin{equation}\label{eq:cobound-exp}
\rho m \cdot \min_{a\in \sC^{m-1}\colon \delta a = c}  \norm{a}  \le \norm{c},    
\end{equation}
where $\delta\colon \sC \to \sC$ is the coboundary map of the complex $\sC = \sC(g_1,\dots,g_m)$.
Here we assume that $n_1 = \dots = n_m$, and we let\footnote{In the general case, we need to use a~more general weight instead of $\wt_X(a)$.} $\norm{a} := \frac{1}{|X(i)|} \cdot \wt_X(a)$ if $a\in \sC^i$.
For example, in the two-dimensional case, we have the complex shown in Fig.~\ref{fg:coboundary}.  Note that formula~\ref{eq:cobound-exp} can be expressed in a~more standard way if we consider the $i$-th Cheeger constant $h^i(\sC)$ for the complex $\sC$ with respect to the norm $\norm{\cdot}$:

\begin{equation*}
h^i(\sC) :=  \min_{x\in\sC^{i}\setminus B^{i}} \frac{\norm{\delta x}}{\min_{b\in B^{i}} \norm{x - b}},  
\end{equation*}
Now if we recall that $\sC$ is a product of acyclic complexes (all homology groups are trivial in (\ref{eq:chain-g})), then by the K\"{u}nneth formula $\sC$ is also acyclic, and hence $B^{m-1} = Z^{m-1}$. Therefore we have  
\begin{equation}
\min_{b\in B^{m-1}} \norm{x - b} = \min_{b\in Z^{m-1}} \norm{x - b} = \min_{a\in \sC^{m-1}\colon \delta a = c} \norm{a},  
\end{equation}
where $c = \delta x$. Thus equation~(\ref{eq:cobound-exp}) can be rewritten as $\rho \le \frac1m h^{m-1}(\sC)$. This means that the product-expansion factor for a~collection of $m$ codes is equal, up to a~normalizing term $1/m$, to the $(m-1)$-th Cheeger constant\footnote{If we view $\sC$ as a~local system on the clique complex $\mathbf{X}(K_{n_1,\dots,n_m})$, then in the chain complex all the dimensions are decreased by one, and $\rho$ corresponds to the $(m-2)$-th normalized Cheeger constant.}:
\begin{equation}
  \rho(\cC_1,\dots,\cC_m) = \frac1m h^{m-1}(\sC).     
\end{equation}

\section{Previous Forms of Product-expansion}\label{app:prod-exp}

In this section, we show that the $\rho$-product-expansion corresponds to a~strong form of the $(s,m,\beta)$-product-expansion from~\cite{Panteleev&Kalachev:stoc2022}, where $s=\Theta(n^2)$, $m=\Theta(n)$. Let us remind this definition from \cite{Panteleev&Kalachev:stoc2022} using the notation of the current work.

A~codeword $x\in \cC=\cC_1\boxplus \cC_2$ is called \emph{$\Delta$-minimal} if the following conditions hold:
\begin{enumerate}
    \item $|x(i,\cdot)|\le d(x(i,\cdot), \cC_2)+\Delta$ for all $i\in [n]$,
    \item $|x(\cdot,j)|\le d(x(\cdot,j), \cC_1)+\Delta$ for all $j\in [n]$,
\end{enumerate}

\begin{definition}
A~pair of codes $(\cC_1, \cC_2)$, $\cC_1,\cC_2\subseteq \F_q^n$, is called \emph{$(s,m,\beta)$-product-expanding} if for each non-zero $\beta n$-minimal 
codeword $x\in \cC_1\boxplus \cC_2$ and for each $A, B\subseteq [n]$ such that 
$|A|,|B|\ge n-m$ we have $|x(A,B)|\ge s$.
\end{definition}

\begin{lemma}
    If a pair of codes $(\cC_1, \cC_2)$, $\cC_1\subseteq \F_q^{n}$, is $\rho$-product-expanding then 
    it is $(s,m,\beta)$-product-expanding, where $s=\rho^2n^2/3$, $m=\rho^2 n/6$, $\beta=\rho/3$.
\end{lemma}
\begin{proof}
    Consider some nonzero $\beta n$-minimal codeword $x\in \cC_1\boxplus \cC_2$. 
    Suppose $|x|< 2s$. In this case by $\rho$-product-expansion we have $x=y+z$ for some $y\in \cC_1\otimes \F_q^{A_2}$, $z\in\F_q^{A_1}\otimes \cC_2$ for some $A_1,A_2\subseteq [n]$ such that $|A_1|+|A_2|\le \frac{|x|}{n\rho} < \frac{2s}{n\rho}=2\rho n/3$. We can assume that each row of $z$ from $A_1$ is nonzero and each column of $y$ from $A_2$ is nonzero, otherwise we can take smaller $A_1$ and $A_2$. 
    Then $\min(|A_1|,|A_2|)< \rho n/3$. Without loss of generality assume that $|A_1|<\rho n/3$. Consider any nonzero column $y(\cdot, i)$. Since $x(\cdot,i)$ differs from $y(\cdot, i)$ only in rows from $A_1$, we have $|x(\cdot,i)-y(\cdot,i)|\le |A_1|< \rho n/3$. Since $y(\cdot,i)\in \cC_1$, we have 
    $$|x(\cdot,i)|\ge \underbrace{|y(\cdot,i)|}_{\ge \rho n}-\underbrace{|x(\cdot,i)-y(\cdot,i)|}_{<\rho n/3}\ge 2\rho n /3>|x(\cdot,i)-\underbrace{y(\cdot,i)}_{\in \cC_1}|+\underbrace{\rho n/3}_{\beta n},$$
    and we have a contradiction with $\beta n$-minimality of $x$. 
    Hence, $|x|\ge 2s$, therefore for each $A,B\subseteq [n]$ such that $|A|,|B|\ge n-m$ we have $|x(A,B)|\ge |x|-(n^2-|A\times B|)\ge |x|-n^2+(n-m)^2\ge 2s - 2mn = s$.
    %
    %
    %
\end{proof}

\begin{lemma}
    If a pair of codes $(\cC_1,\cC_2)$ is $(s,m,\beta)$-product-expanding, then it is $\rho$-product-expanding with $\rho=\min(\frac{s}{2n^2},\beta)$.
\end{lemma}
\begin{proof}
    For a codeword $x\in \cC_1\boxplus \cC_2$ we can define the value
    $$\phi(x):=\min\fbr{|A_1|+|A_2|\mid A_1,A_2\subseteq [n] : x\in \cC_1\otimes \F_q^{A_2}+\F_q^{A_1}\otimes \cC_2}.$$
    Let us prove by induction on the weight of $x$ that for all $x\in \cC_1\boxplus \cC_2$, $|x|<s$, that the following condition holds \begin{equation}\label{eqn:rho-exp}\phi(x)\le \frac{|x|}{n\rho}.\end{equation} 
    For a~non-zero codeword it is clear that $\eqref{eqn:rho-exp}$ holds. Consider a codeword $x\in \cC_1\boxplus \cC_2$, $|x|< s$. From $(s,m,\beta)$-product-expansion it follows that 
    the weight of $x$ can be decreased by more that $\beta n$ by adding either a column from $\cC_1$ or a row from $\cC_2$. 
    Without loss of generality we assume that it is the column $x(\cdot,i)$. Hence there exists $t\in \cC_1$ such that  $|x(\cdot,i)-t|<|x(\cdot,i)|-\beta n$. Now we define  
    $$x'(j,k)=\begin{cases}x(j,k),\mbox{ if }k\ne i,\\t_j,\mbox{ if }k=i.\end{cases}$$
    
    Hence we get that $x'\in \cC_1\boxplus \cC_2$, and $x'$ differs from $x$ only in the $i$-th column. Therefore we obtain  $|\phi(x)-\phi(x')|\le 1$.
    But since $|x'|<|x|$, by the induction hypothesis we get $\phi(x')\le\frac{|x'|}{n\rho}$. Hence we have 
    $$\phi(x)\le\phi(x')+1\le \frac{|x'|}{n\rho}+1\le \frac{|x|-\beta n+\rho n}{n\rho}\le \frac{|x|}{n\rho}.$$
    Now it remains to notice that by definition we have $\phi(x)\le 2n$, and therefore for $|x|\ge s$ we obtain  $\phi(x)\le 2n\le \frac{|x|}{n\cdot s/(2n^2)}\le\frac{|x|}{\rho n}$, i.e., condition~\eqref{eqn:rho-exp} holds and the pair of codes  $(\cC_1,\cC_2)$ is $\rho$-product-expanding.
\end{proof}

\section{Upper Bound on Product Expansion Factor}\label{app:upper}
\begin{proposition}\label{prop:upper}
If $(\cC_1,\cC_2)$ is $\rho$-product-expanding, where~$\cC_i$ is a~linear $[n, k\ge (1-\eps_i)n, d \ge \delta n]_q$ code. Then we have 
\begin{equation*}
    \rho\le \eps_1\eps_2+1/n.
\end{equation*}
\end{proposition}
\begin{proof}
    \begin{figure}\centering
    \begin{tikzpicture}
    \begin{scope}
    \draw[red,dashed] (0,1) -- (2,1) -- (2,3);
    \draw[thick, red] (0,3) -- (2,1);
    \draw[] (0,0) rectangle (3,3);
    \draw[black,dashed] (2,0) -- (2,1);
    \draw[red,->] (1,1) -- (1,0.5);
    \draw [decorate,decoration={brace,amplitude=8pt}] (0,1) -- (0,3) node [midway,left,xshift=-0.8em] {$A_1$};
    \draw [decorate,decoration={brace,amplitude=8pt}] (0,3) -- (2,3) node [midway,above,yshift=0.8em] {$A'_2$};
    \end{scope}
    \begin{scope}[shift={(5,0)}]
    \node[above] at (1.5,3) {$x$};
    \draw[->,>=stealth] (-1.5,1.5) -- node[above]{encode $\cC_1$}(-0.5,1.5);
    \draw[thick] (0,3) -- (2,1);
    \fill[gray!30] (0,0) rectangle (2,1); 
    \fill[red!30] (0,0) rectangle (1.5,1);
    \draw[red,dashed] (0,1) -- (1.5,1) -- (1.5,0);
    \draw[black,dashed] (1.5,1) -- (3,1);
    \draw[red, ->] (1.5,0.5) -- (2.2,0.5);
    \draw[] (0,0) rectangle (3,3);
    \draw [decorate,decoration={brace,amplitude=8pt}] (1.5,0) -- (0,0) node [midway,below,yshift=-0.8em] {$A_2$};
    \end{scope}
    \begin{scope}[shift={(10,0)}]
    \node[above] at (1.5,3) {$y$};
    \draw[->,>=stealth] (-1.5,1.5) -- node[above]{encode $\cC_2$}(-0.5,1.5);
    \fill[gray!30] (1.5,0) rectangle (3,1); 
    \draw (0,0) rectangle (3,3);
    \draw[thick] (0,3) -- (2,1);
    \draw [decorate,decoration={brace,amplitude=8pt}] (3,0) -- (1.5,0) node [midway,below,yshift=-0.8em] {$\eps_2 n$};
    \draw [decorate,decoration={brace,amplitude=8pt}] (1.5,0) -- (0,0) node [midway,below,yshift=-0.8em] {$A_2$};
    \draw [decorate,decoration={brace,amplitude=8pt}] (3,1) -- (3,0) node [midway,right,xshift=0.8em] {$\eps_1 n$};
    \draw [decorate,decoration={brace,amplitude=8pt}] (3,3) -- (3,1) node [midway,right,xshift=0.8em] {$A_1$};
    \end{scope}
    \end{tikzpicture}
    \caption{Codeword construction for Proposition~\ref{prop:upper}. We start with a~diagonal matrix on the subset $A_1\times A'_2$ where $A'_2$ contains an~information set $A_2$ of $\cC_2$ (the left picture). First, we encode it by columns and obtain $x\in\cC_1\otimes \F_q^n$ (middle picture). Then we encode part $x([n]\setminus A_1,A_2)$ by rows, subtract the result from the word $x$, and obtain $y\in \cC_1\boxplus \cC_2$ (the right picture). }\label{fig:prop-proof}
    \end{figure}
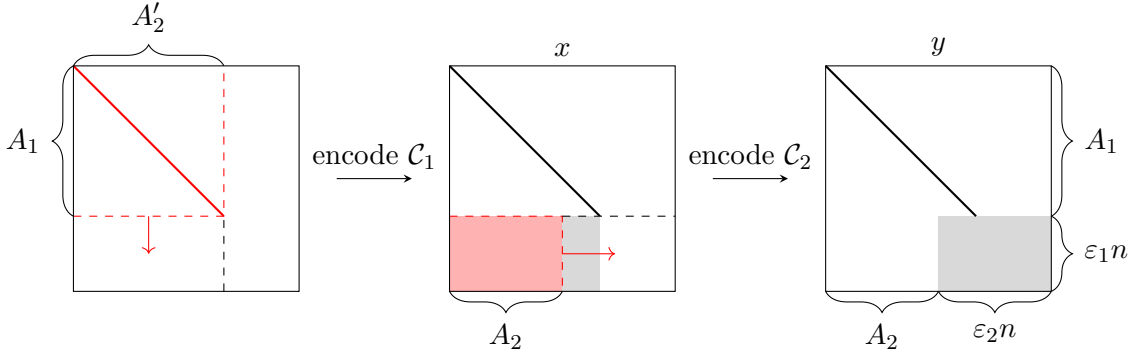
    Suppose the pair $(\cC_1,\cC_2)$ is $\rho$-product-expanding.    Let $A_i$ be an information set of the $\cC_i$, $i\in [2]$. Without loss of generality we can assume that $\eps_1\le\eps_2$, i.e. $|A_1|\ge |A_2|$. 
    Let $A'_2$ be a set of size $|A_1|$ that contains $A_2$. For a set $A$ by $\bar A$ we denote the set $[n]\setminus A$. 
    Consider word $y\in C_1\boxplus C_2$ obtained as shown in figure \ref{fig:prop-proof}. By construction we have $|y|\le \eps_1\eps_2n^2+n$.
    Suppose $y=y_1+y_2$ where $y_1\in \cC_1\otimes \F_q^{B_2}$, $y_2\in \F_q^{B_1}\otimes \cC_2$ such that $|y|\ge \rho n(|B_1|+|B_2|)$. Let us estimate $|B_1|+|B_2|$. It is easy to see that to cover diagonal part $y(A_1,A'_2)$ we need at least $|A_1|$ rows and columns, hence $|B_1\cap A_1|+|B_2\cap A'_2|\ge |A_1|=(1-\eps_1)n$.
    
    Let $\alpha:=\frac1n(|B_1\setminus A_1|+|B_2\setminus A'_2|)$. Then $\alpha\le \frac1n(|\bar A_1|+|\bar A'_2|)=2\eps_1$, $|y|= |A_1|+|y(\bar A_1,\bar A_2)|\le n+\min(\alpha n\cdot \eps_2 n,\eps_1\eps_2 n^2)$. Then we have 
    \[
        \rho\le \frac{|y|}{n(|B_1|+|B_2|)}\le \frac{n+\eps_2 \min(\alpha,\eps_1) n^2}{n(n-\eps_1 n+\alpha n)}\le \frac{1}{n}+\overbrace{\frac{\eps_2 \min(\alpha,\eps_1)}{1-\eps_1+\alpha}}^{\mbox{maximal when }\alpha=\eps_1}\le\frac1n +\eps_1\eps_2.\qedhere
    \]
\end{proof}

From the proof it is easy to see that this bound is not tight except for the degenerate case $\eps_1\eps_2=0$ since we used not very accurate bounds for the weight of the diagonal part of the matrix. Moreover, for large fields, we can use different field elements on the diagonal to make some additional elements in $y(\bar A_1,\bar A_2)$ equal to zero. However, it is not possible to obtain in this way bounds lower than $\eps_1\eps_2$.

\end{document}